\UseRawInputEncoding

\documentclass[12pt]{article}
\usepackage[utf8]{inputenc}

\usepackage[margin=1in, letterpaper]{geometry}
\usepackage{xcolor}
\usepackage{amsmath, amssymb}

\usepackage{graphicx}
\usepackage{mathtools}

\usepackage{amsthm}
\usepackage{authblk}
\usepackage{multirow}
\usepackage{booktabs}
\usepackage{subcaption}
\usepackage{enumitem}
\usepackage{verbatim}

\newtheorem{theorem}{Theorem}
\newtheorem{lemma}{Lemma}
\newtheorem{definition}{Definition}
\newtheorem{remark}{Remark}
\newtheorem{proposition}{Proposition}
\newtheorem{assumption}{Assumption}

\usepackage[
  backend=biber,
  style=authoryear,
  sorting=nyt,
  maxcitenames=3,
  mincitenames=1,
  maxbibnames=6,
  giveninits=false,
  doi=false,
  url=false,
  isbn=false,
  eprint=false,
  uniquename=false
]{biblatex}

\usepackage[colorlinks=true, allcolors=blue]{hyperref}

\title{Detecting Early and Late Divergences in Survival Curves Using Nonparametric Effect Measures}
\date{}
\author[1]{Patrick B. Langthaler$^{*,}$}
\author[2]{Jun Ma}
\author[3]{Jonas Beck$^{*, }$}

\affil[1]{Intelligent Data Analytics (IDA) Lab, Department of Artificial Intelligence and Human Interfaces (AIHI), Paris Lodron University of Salzburg, Salzburg, Austria}
\affil[2]{School of Mathematical and Physical Sciences, Macquarie University, Sydney, Australia}
\affil[3]{Division of Biostatistics, German Cancer Research Center, Heidelberg, Germany}

\vspace{2ex}
\affil[*]{These authors contributed equally and share first authorship.}

\begin{document}
\maketitle

\newpage
\begin{abstract}
Clinical trials often show treatment curves that diverge early and converge later, or vice versa patterns that are poorly captured by the proportional-hazards assumption. We develop a joint inferential framework for two nonparametric functionals of censored survival data: the Kaplan--Meier-based Mann--Whitney effect and a novel temporal contrast separating early and late differences. The approach provides interpretable, probability-scale effect measures and enables joint inference for global and temporal contrasts under right censoring. In simulation studies, the method outperforms the log-rank test under non-proportional hazards while maintaining nominal type-I error. A real-world application illustrates how the temporal contrast reveals clinically meaningful early treatment advantages that remain hidden in standard analyses


\end{abstract}

\newpage
\section{Introduction}
Comparing time-to-event distributions between two groups is one of the central tasks in clinical trial analysis. The standard approach is based on the log-rank test, often accompanied by a Cox proportional hazards model and a hazard-ratio summary. This framework is attractive when the treatment effect is approximately constant over time, and it remains the default approach in many randomized clinical trials. However, the proportional-hazards assumption is restrictive and may be violated in clinically relevant settings. In oncology and immunotherapy trials, treatment effects may be delayed, may diminish over time, or may even change direction during follow-up. In such situations, Kaplan--Meier curves may cross or show early separation followed by convergence, and a single hazard ratio can be difficult to interpret (\cite{royston24,Ananthakrishnan2021CriticalRO, pak17,mukhopadyay22}).

Non-proportional hazards create both inferential and interpretational challenges. The log-rank test is optimal under proportional hazards, but may lose power when early and late deviations have opposite signs and cancel, as can occur with crossing survival or hazard curves.
Several reviews and comparative studies have emphasized that crossing survival curves are not rare in clinical applications and that the standard log-rank test is often still used in such settings, despite its potential loss of power (\cite{Li2015StatisticalIM,dormuth22}). This has motivated a broad literature on methods for non-proportional hazards.

Testing-focused approaches include weighted log-rank procedures, such as Gehan-Wilcoxon, Peto-Peto, and Fleming-Harrington-type tests (\cite{wugilbert22,karrison16}), as well as more flexible adaptive versions and MaxCombo-type procedures
(\cite{mukhopadyay22,lin23,flandre20,yangprentice09}).
Other methods target crossing or more general departures from equality, including squared observed-minus-expected deviations(\cite{linwang04}), two-stage (\cite{qiusheng08}) and partitioned log-rank tests (\cite{LiuYin2017Partitioned}), tests for crossing hazard-rate curves (\cite{zhang24}), permutation tests (\cite{Wyupek2021APT}), omnibus divergence methods (\cite{you23}), and area-based procedures comparing survival curves (\cite{linxu10,Huang02092022,liu20}). Comparative studies and user guidelines
suggest that such methods can be useful under crossing or otherwise
non-proportional hazards, but also that their performance depends strongly on
the underlying alternative pattern and on whether the aim is testing,
estimation, or clinical interpretation (\cite{Li2015StatisticalIM,dormuth22}).

A complementary literature emphasizes interpretable estimands. Restricted mean survival time (RMST) summarizes average event-free survival up to a prespecified horizon without requiring proportional hazards and has been proposed for both trial design and analysis
(\cite{royston13,pak17}). Related summaries include life-expectancy differences
and ratios (\cite{Dehbij2250}), long-term average hazard measures for delayed
immunotherapy effects (\cite{horiguchi25}), and RMST-based survival-benefit
summaries in health technology assessment (\cite{monnickendam19}). Pairwise comparison methods provide probability-scale treatment effects through win probabilities, win ratios, and generalized pairwise comparisons, including stratified (\cite{gasparyan20}), adjusted, competing-risk (\cite{cantagallo21}), and censoring-aware extensions (\cite{mao24,Wu2026WinRA}).
These estimand-based methods are clinically attractive, but a single global RMST difference, win ratio, or pairwise probability may still average over opposing early and late effects.

Existing approaches can therefore be viewed as either testing-focused or estimand-focused. The former aim to improve detection under non-proportional hazards, whereas the latter provide interpretable one-dimensional treatment-effect summaries. Yet a gap remains: omnibus tests usually do not describe how effects change over time, while one-dimensional estimands may obscure opposing early and late effects.

We address this gap by proposing a two-dimensional nonparametric effect vector
for right-censored two-sample survival data. The first component is a
Kaplan--Meier-based Mann--Whitney effect for the global pairwise ordering of
event times; the second is a median-split temporal contrast comparing early and
late pairwise ordering. Thus, the proposed framework targets interpretable
early-late survival patterns rather than arbitrary distributional differences.
We develop estimation, joint asymptotic inference, bootstrap-based Wald-type and
max-t-type tests, and simultaneous confidence intervals, and evaluate the method
in simulations and data examples.
\subsection*{Motivating Example}
%As a data example we use data from a randomized phase III trial testing the efficacy of whole-brain radiation therapy (WBRT) in addition to local treatment vs local treatment alone (observation) for melanoma brain metastases. Hong et al~\cite{hong2019adjuvant} recruited 215 patients between April 2009 and September 2017. In total 107 patients were assigned to the WBRT arm and 108 to the observation arm. Of those 100 and 107 respectively were available for the final analysis. The median follow-up time was 48.1 months. Here we focus on the outcome of overall mortality over the entire study period. Kaplan-Meier plots of both groups can be seen in Figure~\ref{fig:km_wbrt}. Clearly we have a crossover of the two survival curves at about 27 months after the start of the observation period, with WBRT seemingly performing better before and worse afterward. Unsurprisingly the logrank test does not reject at and alpha of $0.05$, indeed giving a p-value of $0.83$.  The Gehan-Wilcoxon test, the Yang and Prentice test and the two-stage procedure give p-values of $0.44$, $0.64$ and $0.14$ respectively. Our procedure gives a p-value of $0.11$ with the max-t-type test for the pooled bootstrap. As an estimate for $(\theta, O)$ we get $(0.52, 0.40)$, suggesting that it is mainly the deviation of the overlap index from $0.5$ that is responsible for the low p-value.
We illustrate the proposed method using data from the randomized phase III trial
by ~\cite{hong2019adjuvant}, which compared adjuvant whole-brain radiation
therapy (WBRT) with observation after local treatment of one to three melanoma
brain metastases. The original trial was designed primarily to assess distant
intracranial failure within 12 months. Here, we restrict attention to overall
survival, defined as the time from randomization to death from any cause, with
patients alive at last follow-up treated as right-censored.

In the original analysis, no significant overall-survival difference was found
between the two groups. The published median overall survival was 16.5 months
in the WBRT group and 13 months in the observation group, and the published
Kaplan--Meier analysis yielded a log-rank \(p\)-value of \(0.89\). Thus, a
standard global comparison did not provide evidence for an overall survival
benefit of WBRT.

For the present reanalysis, we used all patients with complete information on
randomization date, survival status, and either date of death or last follow-up.
This yielded 107 patients in the observation arm and 99 patients in the WBRT
arm. Death was observed in 75 patients in each group. The Kaplan--Meier curves
are shown in Figure~\ref{fig:km_wbrt}. They suggest a non-proportional-hazards pattern:
the WBRT curve is slightly above the observation curve during the earlier part
of follow-up, whereas the ordering appears to attenuate and eventually reverse
later in follow-up.

\section{Statistical Model}

We consider two independent groups, for example a treatment group and a
control group. In the first group, let
$ %\[
T_i \sim F %\qquad 
$ %\]
denote the event times, and 
$%\[
C_i \sim K %\qquad 
$%\] 
 denote the corresponding censoring times, where $i=1,\ldots,n_1$. In the second group, let
$%\[
U_j \sim G  %\qquad 
$ %\]
denote the event times, and
$%\[
D_j \sim J  %\qquad 
$ %]
denote the corresponding censoring times, where $j=1,\ldots,n_2$. The distribution functions \(F\)
and \(G\) are the primary objects of interest, whereas \(K\) and \(J\)
describe the censoring mechanisms.

We assume that
\[
(T_1,C_1),\ldots,(T_{n_1},C_{n_1})
\quad\text{and}\quad
(U_1,D_1),\ldots,(U_{n_2},D_{n_2})
\]
are independent samples, that the two groups are mutually independent, and
that censoring is independent within each group, that is,
$%\[
T_i \perp C_i,
%\qquad
U_j \perp D_j .~
$%]
The observed data are
$ %\[
X_i = \min(T_i,C_i), 
%\qquad
\Delta_i = \mathbf 1\{T_i \leq C_i\}, 
%\qquad 
i=1,\ldots,n_1,
$ %\]
and
$ %\[
Y_j = \min(U_j,D_j), 
%\qquad
\Gamma_j = \mathbf 1\{U_j \leq D_j\}, 
%\qquad 
j=1,\ldots,n_2.
$ %\]

Our aim is to compare the two event-time distributions \(F\) and \(G\)
nonparametrically. Instead of imposing a proportional-hazards model, we
summarize their difference by the bivariate effect parameter
\begin{equation} \label{eq1}
%\[
\psi(F,G)
= (
%\begin{pmatrix}
\theta(F,G),~ %\\
O(F,G) )^\top
%\end{pmatrix},
%\] 
\end{equation}
where \(\theta(F,G)\) is a Mann--Whitney-type relative effect and
\(O(F,G)\) is a median-split temporal contrast designed to distinguish
early and late distributional differences. The precise definitions of
these two functionals are given below.

\subsection{Nonparametric effects}

For any c\`adl\`ag cdf $A$, write the left limit as $A(t-)$ and define the mid-cdf
\[
A^{\pm}(t):=\frac{A(t-)+A(t)}{2}.
\]
%This convention ensures identities such as $\theta(F,F)=1/2$ even in the presence
%of jumps/ties; if the event-time distribution is continuous, then $A^{\pm}(t)=A(t)$.

\begin{definition}
    The relative effect (Mann--Whitney effect) is defined as
\begin{equation}
\label{eq:theta-def}
\theta(F,G):=\int_{[0,\infty)} F^{\pm}(t)\,dG(t),
\end{equation}
where the integral is understood in the Lebesgue--Stieltjes sense.
\end{definition}
For $T\sim F$ and $U\sim G$ being independent, we get
\begin{equation}
\label{eq:theta-interpret}
\theta(F,G)=\mathbb{P}(T<U)+\tfrac12\,\mathbb{P}(T=U),
\end{equation}
and in the continuous case $\theta(F,G)=\mathbb{P}(T\le U)$ $= \mathbb{P}(T < U)$.
Values $\theta(F,G)>1/2$ indicate that $T$ tends to be smaller than $U$
(earlier event times for $F$ than for $G$).
\begin{definition} \label{definition2}
    Let $m$ be a median of $G$, i.e., any $m\in\arg\min_t |G(t)-1/2|$ (equivalently,
a version of $G^{-1}(1/2)$). Define the early/late conditional relative effects
$ %\[
\theta_{\le}(F,G):=\mathbb{E}\!\left[F^{\pm}(U)\mid U\le m\right], 
%\qquad
\theta_{>}(F,G):=\mathbb{E}\!\left[F^{\pm}(U)\mid U> m\right], 
%\quad 
(\text{where } U\sim G),
$ %\]
which can be written as %\textcolor{red}{(however $T\geq 0$!)}
\begin{equation}
\label{eq:O-partial_def}
\theta_{\le}(F,G)=\frac{\int_{(0,m]} F^{\pm}(t)\,dG(t)}{G(m)},
\qquad
\theta_{>}(F,G)=\frac{\int_{(m,\infty)} F^{\pm}(t)\,dG(t)}{1-G(m)}.
\end{equation}
The Overlap Index (median-split temporal contrast) is then
\begin{equation}
\label{eq:O-def}
O(F,G):=\theta_{>}(F,G)-\theta_{\le}(F,G).
\end{equation}
\end{definition}

Thus, \(O(F,G)\) measures the separation between the lower and upper halves of the \(G\)-distribution on the \(F\)-rank scale, with reference value \(1/2\) under \(F=G\).
Under the assumption of continuous distributions it can be shown that Equation \eqref{eq:O-def} is equivalent to the definition of the Overlap Index in \cite{Beck2023Combining}.

%All theoretical statements are formulated on a fixed interval $[0, \tau]$, chosen a priori and contained in the identifiable follow-up region of both groups. Equivalently, the event times are replaced by \(T^\tau=T\wedge\tau\) and \(U^\tau=U\wedge\tau\) 
%\textcolor{red}{(but they were called $X$ and $Y$ previously)}. \textcolor{teal}{Yes I think it should be \(X^\tau=X\wedge\tau\) and \(Y^\tau=Y\wedge\tau\). Maybe we should not speak of event times in this case but observed times of something. I think this part is a bit problematic in general and we need to be very careful with this. For the proofs we do require that $\tau$ is fixed and not some $\hat{\tau}$ estimated from the data. However for our estimator $\hat{\theta}$ we do estimate $\tau_X := \max X_i$ from the data and this then takes over the role of $\tau$ in the proof. On the other hand $\tau_X$ does not really estimate anything and our proofs hold or arbitrary $\tau$ so maybe this is not a problem.}

%To keep the notation simple, we continue to write \(F\), \(G\), \(\theta(F,G)\), \(O(F,G)\), and \(\psi(F,G)\) for the corresponding restricted quantities. Thus, all integrals below are understood over \([0,\tau]\), unless stated otherwise.

\begin{remark}
The median split used in Definition \ref{definition2} is not essential. More generally, the split may be
placed at any prespecified quantile $p\in(0,1)$ of the distribution $G$, yielding an
analogous quantile-split temporal contrast. We use the median because it provides a balanced
early--late comparison. Other fixed choices of $p$ may nevertheless be useful when there is prior
clinical or biological motivation to focus on an earlier or later part of follow-up.
Further details and the corresponding regularity conditions are provided in the
Supplementary Material \ref{remarks}.
\end{remark}

\subsection{Interpretation and Special Cases}
The vector
$%\[
    \psi(F,G)~
    %=(%\{
    %\theta(F,G),O(F,G)%\}
    %)^{\top}
$%\]
in equation (\ref{eq1}) 
should be interpreted as a two-dimensional probability-scale summary of the
difference between two event-time distributions.
Throughout this section, $F$ denotes the control distribution and $G$ the treatment distribution. The event of interest is assumed to be adverse, such as death, progression, or relapse; for favorable events, the directional interpretation is reversed.

The first component $ \theta(F,G)$
is a global pairwise ordering probability. Values \(\theta(F,G)>1/2\) indicate that event times from $G$ tend to be larger than event times from $F$, and therefore suggest a global survival advantage of the treatment group.
Values
smaller than \(1/2\) indicate the opposite.
Thus \(\theta(F,G)\) provides the primary summary of the overall ordering of the two event-time distributions.

The second component $O(F,G)$ describes how this ordering is distributed over time.
It compares the relative position of the lower and upper halves of the $G$ distribution on the $F$-rank scale and can therefore be interpreted as an early-late contrast of the pairwise event-time ordering.

Figure \ref{subfig:survival-patterns} and \ref{subfig:effect-space} illustrate the joint interpretation.
When $\theta(F,G)$ is close to $1/2$, the global ordering may be weak because early and late effects cancel; in this case, \(O(F,G)\) can reveal whether one group performs relatively better early and worse later, or vice versa. When $\theta(F,G)$ is far from $1/2$, the main finding is the global advantage or disadvantage captured by $\theta(F,G)$, while $O(F,G)$ refines this interpretation by indicating whether the ordering is concentrated earlier, later, or changes over follow-up.

Importantly, $ O(F,G)$ should always be interpreted jointly with $\theta(F,G)$. Values $O(F,G)\neq 1/2$ do not by themselves imply crossing survival curves or non-proportional hazards. Rather, $O(F,G)$ describes the temporal structure of the pairwise ordering summarized by $\theta(F,G)$. The proposed vector is therefore targeted rather than omnibus: it summarizes two clinically interpretable features, namely the overall pairwise ordering of event times and the early--late change in this ordering. This distinguishes it from weighted log-rank, MaxCombo, two-stage, partitioned log-rank, and area-between-curves procedures, which are primarily designed as testing procedures under non-proportional hazards.

\begin{figure}
    \centering
    \begin{subfigure}[t]{0.5\linewidth}
        \centering
        \includegraphics[width=\linewidth]{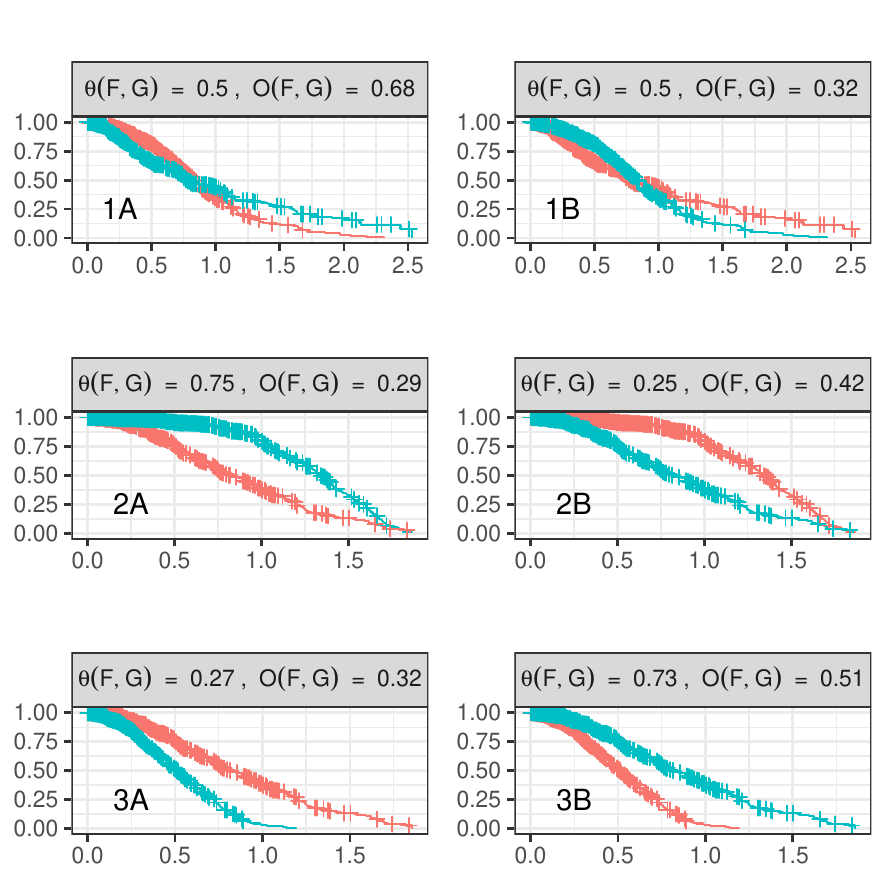}
        \captionsetup{width=1.6\linewidth}
        \caption{\label{subfig:survival-patterns}Illustrative survival-curve patterns and corresponding values of
    \(\theta(F,G)\) and \(O(F,G)\). Here $F$ is the red and $G$ the blue survival curve. The first component \(\theta\) describes the
    global pairwise ordering of event times, whereas \(O\) describes the early--late
    contrast of this ordering. The top row illustrates cases in which global effects
    cancel, \(\theta(F,G)\approx 1/2\), but the temporal contrast separates opposite
    early--late patterns. The middle and lower rows illustrate cases with stronger
    global ordering, where \(O\) refines the interpretation by describing whether
    the difference is concentrated earlier or later in follow-up.}
    \end{subfigure}
    \hfill
    \begin{subfigure}[t]{0.5\linewidth}
        \centering
        \includegraphics[width=\linewidth]{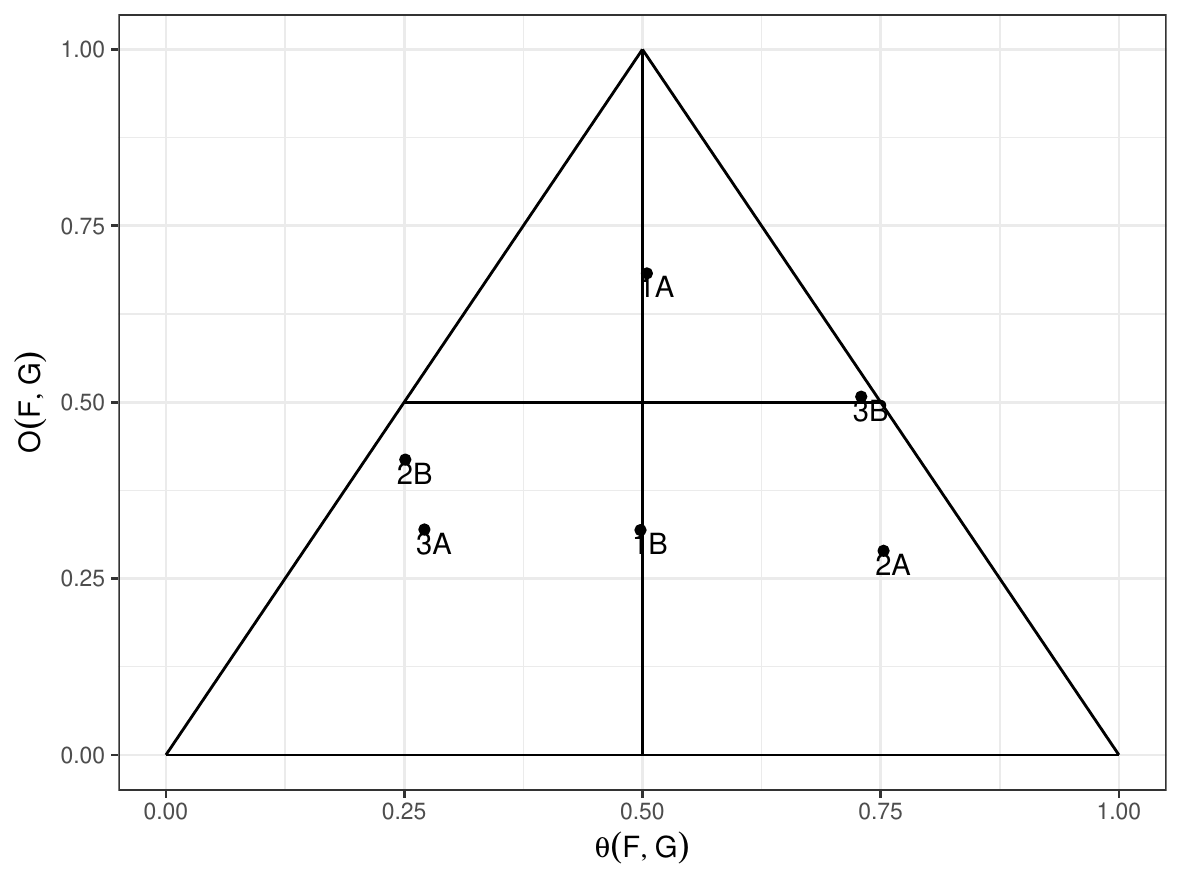}
        \captionsetup{width=1.6\linewidth}
        \caption{\label{subfig:effect-space}Location of the six examples in the
    \((\theta(F,G),O(F,G))\)-plane. The figure illustrates that \(O\) has to be
    interpreted jointly with \(\theta\): its possible range is largest near
    \(\theta=1/2\) and becomes more restricted as the global pairwise ordering
    becomes stronger.}
    \end{subfigure}
    \caption{Examples of the behaviour of $\theta$ and $O$ for different pairs of survival curves.}
    \label{fig:wbrt-example}
\end{figure}

The following special cases help to clarify the interpretation of the
bivariate effect vector \(\psi(F,G)=(\theta(F,G),O(F,G))^\top\).
In simple one-parameter models, both components are determined by the same
underlying model parameter. Thus, the temporal contrast \(O(F,G)\) does not
represent an additional degree of freedom under such models, but becomes most
informative in genuinely time-varying or non-proportional alternatives.

\begin{remark} \label{remark3}
    Suppose that $F$ and $G$ arise from a proportional-hazards model with
    $%\begin{align*}
        %&
        F(t) = 1 - S_0(t)^{\lambda_1},  %\\
        %&
        G(t) = 1 - S_0(t)^{\lambda_2}, ~  
    $%\end{align*}
where $S_0$ is a continuous survival function. Then 
   $\theta = %\frac{\lambda_1}{\lambda_1 + \lambda_2}
   \lambda_1/(\lambda_1 + \lambda_2)
   = %\frac{1}{\eta + 1}
   1/(\eta + 1)
   $, where $\eta = %\frac{\lambda_2}{\lambda_1}
   \lambda_2/\lambda_1
   $ is the hazard ratio.
   Moreover, the temporal contrast is
\[
    O(F,G)
    =
    \frac{2}{1+\lambda_1/\lambda_2}
    \left(1-2^{-\lambda_1/\lambda_2}\right)
    =
    \frac{2\eta}{1+\eta}
    \left(1-2^{-1/\eta}\right).
\]
Hence, under proportional hazards, the vector
\((\theta(F,G),O(F,G))^\top\) is fully determined by the hazard ratio.
\end{remark}
\begin{proof}
  See Section \ref{proofremark3}
\end{proof}
\subsection{Estimation} \label{section2.3}
Based on the observed samples, let
\(\widehat S_X\) and \(\widehat S_Y\) denote the Kaplan--Meier estimators
of the survival functions corresponding to \(F\) and \(G\), respectively.
We define the associated Kaplan--Meier distribution-function estimators by
$ %\[
\widehat F(t) = 1 - \widehat S_X(t), 
%\qquad
\widehat G(t) = 1 - \widehat S_Y(t). 
$ %\]
Equivalently, if $X_{(1)}\le \cdots \le X_{(n_1)}$ are the ordered observed times with
corresponding event indicators $\Delta_{(i)}$, then
\[
\widehat S_X(t)=\prod_{X_{(i)}\le t}\left(1-\frac{\Delta_{(i)}}{n_1-i+1}\right),
\]
and analogously for $\widehat S_Y$.

%\textcolor{teal}{Alles von hier bis zum Beginn der nächsten Section ist umformuliert. Zu der Notation mit $\tau$ vs $\tau_X$ vs $\tau_Y$: Ich habe jetzt mal $\tau$ als $\tau$ gelassen. $\tau_x$ ist nichts anderes als $X_{(n_1)}$, deswegen habe ich es jetzt mal so bezeichnet. $\tau_Y$ definieren wir, aber benutzen es wenn ich das richtig sehe nie, also brauchen wir es eigentlich gar nicht.}

Let $\Delta\widehat G(t):=\widehat G(t)-\widehat G(t-)$ denote the jump size of
$\widehat G$ at time $t$. 
%\cancel{, and let
%$ %\[
%\tau_Y:=\sup\{t:\Delta\widehat G(t)>0\}
%$ %\]
%be the last jump time.
%}
Note that since $\widehat G$ is a step function, integrals w.r.t.\
$d\widehat G$, and therefore estimators created by plugging $\widehat{F}$ and $\widehat{G}$ into expressions~\ref{eq:theta-def} and~\ref{eq:O-partial_def}, reduce to finite sums over its jump times.

As usual, the Kaplan--Meier estimators are extended as constant
functions beyond the largest observed time in the corresponding
sample. Let $X_{(n_1)}:=\max_{1\le i\le n_1} X_i$ denote the largest observed time in group $X$. Since the target functional satisfies
$\theta(F,G)+\theta(G,F)=1$ and $\theta(F,F)=1/2$, it is desirable that the corresponding estimator obeys the same property. In the present censored-data setting, we are not aware of an existing
estimator that satisfies this condition in finite samples. We therefore enforce it by the estimator
\begin{equation}
\label{eq:theta-star-simplified}
\widehat\theta=\widehat\theta(F,G)
=
\sum_{j:\: Y_j \leq X_{(n_1)}} \widehat{F}^{\pm}(Y_j) \Delta\widehat{G}(Y_j)
+\frac{1+\widehat F(X_{(n_1)})}{2}\bigl(1-\widehat G(X_{(n_1)})\bigr),
\end{equation}
i.e., the KM tail mass $1-\widehat G(X_{(n_1)})$ is downweighted by
$\tfrac{1+\widehat F(X_{(n_1)})}{2} \in[1/2,1]$. It can be easily shown that this estimator is exactly the mean of the two estimators proposed by \cite{wang2009estimation}.
%\textcolor{teal}{Irgendwie fehlt uns jetzt der Übergang von $\theta$, wo wir noch von $0$ bis $\infty$ integrieren, zu dem $\widehat{\theta}$ hier, wo wir das nicht mehr machen. Oder meinst du das passt so?}

Let $\widehat m_G$ denote the KM median of group $Y$,
$ %\[
\widehat m_G := \inf\{t:\widehat S_Y(t)\le 1/2\},
$ %\]
whenever it exists (if $\widehat S_Y$ does not fall below $1/2$ due to heavy
censoring, $O(F,G)$ is not identifiable via a median split and should be treated
as not estimable from the available follow-up). We define normalized
conditional plug-in estimators

%\textcolor{teal}{I reformulated the sets over which we define the sums a bit. I think we should not really sum over $\{t < \widehat{m}_G\}$ since this is not a countable set.}

\[
\widehat\theta_{\le}
=
\frac{\sum_{j:\: Y_j \le \widehat{m}_G}\widehat F(Y_j-)\Delta\widehat G(Y_j)}
     {\sum_{j:\: Y_j \le \widehat{m}_G}\Delta\widehat G(Y_j)},
\qquad
\widehat\theta_{>}
=
\frac{\sum_{j:\: Y_j > \widehat{m}_G}\widehat F(Y_j-)\Delta\widehat G(Y_j)}
     {\sum_{j:\:Y_j > \widehat{m}_G}\Delta\widehat G(Y_j)},
\qquad
\widehat O=\widehat\theta_{>}-\widehat\theta_{\le}.
\]
Using $\widehat F(t-)$ (left limits) aligns the computation with the mid-cdf/tie
convention and avoids ambiguities at jump points.

\subsection{Asymptotics}
For two distribution functions $F$ and $G$, define the bivariate parameter vector
$ %\[
\psi(F,G)
=
%\begin{pmatrix}
(\theta(F,G),~ %\\[0.3em]
O(F,G))^\top.
%\end{pmatrix}
$ %\]
In order to construct hypothesis tests and confidence intervals we study the asymptotic distribution of 
$ %\[
\widehat{\boldsymbol\psi}_N
=
%\begin{pmatrix}
(\widehat \theta(\widehat{F}, \widehat{G}), ~ %\\[0.3em]
\widehat O(\widehat{F}, \widehat{G}))^\top.
%\end{pmatrix}.
$ %\]
We then have the following

\begin{theorem}[Joint asymptotic normality] \label{thm:asym}
Assume independent right censoring within each group and that the regularity conditions in Assumption~\ref{ass:regularity} in the Supplementary Material hold.
Let $N=n_1+n_2$ and assume
$ %\[
%\frac{n_1}{N} 
n_1/N\to \lambda \in (0,1).
$ %\]
 Then
\[
\sqrt N
\left\{
\widehat{\boldsymbol\psi}_N-
\boldsymbol\psi(F,G)
\right\}
\overset{d}{\longrightarrow}
\mathcal N_2
\left(
\boldsymbol 0,
\boldsymbol\Sigma(F,G)
\right).
\]
where $\boldsymbol \Sigma(F,G)$ %and $\Sigma(H)$ 
is a finite positive semidefinite covariance matrix.
\end{theorem}

\begin{proof}
See Section \ref{proofthm1}.
\end{proof}

\subsection{Resampling}

The asymptotic covariance matrix of
\(
\widehat\psi_{N}
\)
can in principle be derived analytically from the functional delta method. Since the resulting expression is cumbersome, we use nonparametric resampling. We distinguish two bootstrap schemes with different purposes: a stratified bootstrap, which estimates the sampling variability under the observed two-sample model, and a pooled bootstrap, which imposes the null hypothesis of a common distribution.

The stratified bootstrap resamples within each treatment group. For
\(b=1,\ldots,B\), draw
\(
(X_1^{*,b},\Delta_1^{*,b}),\ldots,
(X_{n_1}^{*,b},\Delta_{n_1}^{*,b})
\)
independently with replacement from
\(
(X_1,\Delta_1),\ldots,(X_{n_1},\Delta_{n_1}),
\)
and, independently, draw
\(
(Y_1^{*,b},\Gamma_1^{*,b}),\ldots,
(Y_{n_2}^{*,b},\Gamma_{n_2}^{*,b})
\)
with replacement from
\(
(Y_1,\Gamma_1),\ldots,(Y_{n_2},\Gamma_{n_2}).
\)
Let \(\widehat F^{*,b}\) and \(\widehat G^{*,b}\) be the corresponding
Kaplan--Meier estimators, and set
\(
\widehat{\boldsymbol{\psi}}_N^{*,b,\mathrm{str}}
=
\widehat{\boldsymbol{\psi}}
\bigl(\widehat F^{*,b},\widehat G^{*,b}\bigr).
\)
The stratified bootstrap covariance estimator is
\[
\widehat{\boldsymbol{\Sigma}}_{\mathrm{str}}
=
N\operatorname{Cov}_{b=1,\ldots,B}
\left[
\widehat{\boldsymbol{\psi}}_N^{*,b,\mathrm{str}}
\right].
\]

The pooled bootstrap approximates the distribution of the test
statistic under the null hypothesis
\(
H_0:F=G.
\)
Pooling the observed censored data requires exchangeability of the two
samples. Thus, in addition to \(H_0:F=G\), we assume equal censoring
mechanisms, \(K=J\).
For \(b=1,\ldots,B\), draw two independent bootstrap samples with
replacement from the pooled set of observed pairs
\(
\bigl\{(X_i,\Delta_i):i=1,\ldots,n_1\bigr\}
\cup
\bigl\{(Y_j,\Gamma_j):j=1,\ldots,n_2\bigr\}.
\)
More precisely, draw
\(
(Z_{1,1}^{*,b},\eta_{1,1}^{*,b}),\ldots,
(Z_{1,n_1}^{*,b},\eta_{1,n_1}^{*,b})
\)
and, independently,
\(
(Z_{2,1}^{*,b},\eta_{2,1}^{*,b}),\ldots,
(Z_{2,n_2}^{*,b},\eta_{2,n_2}^{*,b})
\)
with replacement from the pooled observations. Here,
\(Z_{g,i}^{*,b}\) and \(\eta_{g,i}^{*,b}\) denote the observed time and
event indicator, respectively, for observation \(i\) in bootstrap
group \(g\in\{1,2\}\) and bootstrap replicate \(b\).

Let \(\widehat H_1^{*,b}\) and \(\widehat H_2^{*,b}\) be the
Kaplan--Meier estimators computed from these two bootstrap samples.
Define
\(
\widehat{\boldsymbol{\psi}}_N^{*,b,\mathrm{pool}}
=
\widehat{\boldsymbol{\psi}}
\bigl(\widehat H_1^{*,b},\widehat H_2^{*,b}\bigr).
\)
The pooled bootstrap covariance estimator is
\[
\widehat{\boldsymbol{\Sigma}}_{\mathrm{pool}}
=
N\operatorname{Cov}_{b=1,\ldots,B}
\left[
\widehat{\boldsymbol{\psi}}_N^{*,b,\mathrm{pool}}
\right].
\]

\begin{theorem} \label{thm_boot}
Assume Assumptions~\ref{ass:regularity} and
\ref{ass:bootstrap} in the Supplementary Material. Then, conditionally on the data
\begin{align*}
&\sqrt{N}\,
\left(\psi \left(\widehat{F}^*, \widehat{G}^* \right) - \psi \left( \widehat{F}, \widehat{G} \right) \right)
\;\xRightarrow{\mathbb P}\;
\mathcal N_2\bigl(0,\Sigma(F,G)\bigr) 
\end{align*}
If, in addition, \(H_0:F=G\) holds and the observed censored pairs are exchangeable across groups, then the pooled bootstrap satisfies conditional on the data
\begin{align*}
&\sqrt{N}\,
\left(\psi \left( \widehat{H}_1^*, \widehat{H}_2^* \right) - \psi_0 \right)
\;\xRightarrow{\mathbb P}\;
\mathcal N_2\bigl(0,\Sigma_0\bigr)
\end{align*}
where \(\psi_0=(1/2,1/2)^\top\) and \(\Sigma_0\) denotes the asymptotic covariance matrix under the null model.
\end{theorem}
\begin{proof}
See Section \ref{proofthm2}.
\end{proof}

\subsection{Joint Hypothesis Testing}
We consider testing the hypothesis $H_0:\: F = G $.  Additionally we assume equal censoring distributions $K = J$.
Under this hypothesis we have exchangeability, which justifies the use of the pooled bootstrap. We first formulate a test via the classic Wald-type statistic, as well as a max-t-type test. Theorem~\ref{thm_boot} justifies the use of the bootstrapped covariance estimators $\widehat\Sigma_{\mathrm{str}}$ and $\widehat\Sigma_{\mathrm{pool}}$ of $\theta(F, G)$ and $O(F, G)$. Let additionally $\widehat{R}_{\mathrm{pool}}$ and $\widehat{R}_{\mathrm{str}}$ denote the correlation matrices corresponding to $\widehat{\Sigma}_{\mathrm{pool}}$ and $\widehat{\Sigma}_{\mathrm{str}}$, respectively.

\subsubsection{Wald-type Test}

In general the Wald statistic is defined as
\begin{equation*}
W_{N}(\psi_0)
=
N(\widehat\psi_{N}-\psi_0)^\top
(\widehat\Sigma_{\mathrm{pool}})^{-1}
(\widehat\psi_{N}-\psi_0).
\end{equation*}
If \(\widehat\Sigma_{\mathrm{pool}}\) is singular or numerically ill-conditioned in finite samples, we replace the ordinary inverse by the Moore--Penrose generalized inverse.
Under $H_0$ we have
$ %\begin{equation*}
\psi(F, G) = \psi_0 = (1/2, 1/2)
$ %\end{equation*}
and thus the Wald-type statistic becomes, provided that the limiting covariance matrix
\(\Sigma_0\) is positive definite
\begin{equation*}
W_{N}\!\left(\tbinom{1/2}{1/2}\right)
=
N\left(\widehat\psi_{N}-\tbinom{1/2}{1/2}\right)^\top
(\widehat\Sigma_{\mathrm{pool}})^{-1}
\left(\widehat\psi_{N}-\tbinom{1/2}{1/2}\right).
\end{equation*}
Under $H_0$, and provided that $G$ is continuous, we have
$ %\begin{equation*}
W_{N}(\psi_0)
\xrightarrow{d}
\chi^2_2,
$ %\end{equation*}
so that $H_0$ is rejected at level $\alpha$ if
$ %\begin{equation*}
W_{N}(\psi_0) > q_{\chi^2_2, 1 - \alpha}.
$ %\end{equation*}

\subsubsection{Max-t-type Test}
As an alternative to the quadratic Wald test, we consider a
maximum-t-type test that yields simultaneous inference with
coordinatewise control.

Let $\widehat{\sigma}^{*}_{\theta}$ and $\widehat{\sigma}^{*}_{O}$ be the bootstrapped standard deviations (using the pooled bootstrap) of $\theta$ and $O$ respectively. Define the statistics
\begin{equation*}
Z_1 := \frac{\widehat{\theta}(\widehat{F}, \widehat{G}) - 1/2}{\widehat{\sigma}^{*}_{\theta}}, \qquad
Z_2 := \frac{\widehat{O}(\widehat{F}, \widehat{G}) - 1/2}{\widehat{\sigma}^{*}_{O}}
\end{equation*}
Under $H_0$ we have
$ %\begin{equation*}
\left( Z_1, Z_2 \right) \overset{d}{\rightarrow} \mathcal{N}_{2}(0, R_{\mathrm{pool}}).
$ %\end{equation*}
Define the maximum-t-type statistic as
\begin{equation*}
    T^{\max}_{N} := \max\{|Z_1|, |Z_2|\}
\end{equation*}
and define $c_{1 - \alpha}(R_{\mathrm{pool}})$ as the $1 - \alpha$-Quantile of
    $\max\{|Z'_1|, |Z'_2|\}$ for some  $\left( Z'_1, Z'_2 \right) \sim \mathcal{N}_{2}(0, R_{\mathrm{pool}})$
Then we reject $H_0$ if
$ %\begin{equation*}
T^{\max}_{N} > c_{1 - \alpha}
$. %\end{equation*}
In practice $c_{1 - \alpha}(\widehat{R}_{\mathrm{pool}})$ is computed via Monte-Carlo sampling of $\left( Z'_1, Z'_2 \right)$. 
%Define the maximum-type statistic
%\[
%T_{\max}
%=
%\max_{k=1,2} |Z_k|.
%\]

%Under $H_0$,
%\[
%T_{\max}
%\xrightarrow{d}
%\max\bigl(|Z_1|,|Z_2|\bigr),
%\qquad
%(Z_1,Z_2)^\top \sim \mathcal N_2(0,R).
%\]

%Let $c_{1-\alpha}(R)$ denote the $(1-\alpha)$ quantile of
%$\max(|Z_1|,|Z_2|)$ under a bivariate normal distribution with
%correlation matrix $R$.
%In practice, we use the plug-in estimate $c_{1-\alpha}(\widehat R)$.

%We reject $H_0$ at level $\alpha$ if
%\[
%T_{\max}
%>
%c_{1-\alpha}(\widehat R).
%\]

%For testing distributional equality,
%\[
%H_0:\ (\theta,O)^\top = \bigl(1/2,1/2\bigr)^\top,
%\]
%the same max-$t$ statistic is used with
%$\psi_{0,1}=1/2$ and $\psi_{0,2}=1/2$.
\subsection{Joint Confidence Intervals}
\subsubsection{Wald-type confidence region}
A Wald-type $(1-\alpha)$ confidence region for
$\psi(F,G)=(\theta,O)^\top$ is given by
\[
\mathcal C_{1-\alpha}
=
\left\{
\mathbf d\in\mathbb R^2:
N(\widehat\psi_{N}-\mathbf d)^\top
(\widehat\Sigma_{\mathrm{str}})^{-1}
(\widehat\psi_{N}-\mathbf d)
\le
q_{\chi^2_2,1-\alpha}
\right\}.
\]

Projection yields simultaneous marginal confidence intervals:
\[
\theta(F,G)
\in
\left[
\widehat\theta
\pm
\sqrt{\frac{q_{\chi^2_2}(1-\alpha)}{N}}
\sqrt{e_1^\top \widehat\Sigma_{\mathrm{str}} e_1}
\right],
\qquad 
O(F,G)
\in
\left[
\widehat O
\pm
\sqrt{\frac{q_{\chi^2_2}(1-\alpha)}{N}}
\sqrt{e_2^\top \widehat\Sigma_{\mathrm{str}} e_2}
\right],
\]
where $e_1=(1,0)^\top$ and $e_2=(0,1)^\top$.
These intervals are conservative simultaneous intervals induced by
the bivariate confidence ellipse.

\subsubsection{Max-t-type Simultaneous Confidence Intervals}
The max-$t$ procedure yields simultaneous $(1-\alpha)$ confidence
intervals for the components of $\psi(F,G)$:
%\[
%\psi_k(F,G)
%\in
%\left[
%\widehat\psi_k
%\pm
%\frac{c_{1-\alpha}(\widehat R)\,\widehat\sigma_k}{\sqrt{N}}
%\right],
%\qquad k=1,2.
%\]

\begin{align*}
\theta(F,G) \in \left[ \widehat\theta \pm \frac{c_{1-\alpha}(\widehat{R}_{\mathrm{str}})\, \widehat{\sigma}_{\mathrm{str},\theta}}{\sqrt{N}} \right],  \qquad\text{and}\qquad
O(F,G) \in \left[ \widehat O \pm \frac{c_{1-\alpha}(\widehat{R}_{\mathrm{str}})\, \widehat{\sigma}_{\mathrm{str},O}}{\sqrt{N}} \right]
\end{align*}

These intervals control the familywise error rate asymptotically
and are typically less conservative than projection intervals
obtained from the quadratic Wald ellipse.

\section{Simulations}
We conducted a simulation study to evaluate the finite-sample performance of the proposed
procedures under proportional and non-proportional hazards and to compare them with
established two-sample survival tests. We considered five basic settings: a null scenario,
a proportional-hazards scenario, a crossing-in-the-middle scenario, a crossing-at-the-end
scenario, and a diffuse-transition alternative in which the treatment effect changes gradually
from an early benefit to a later disadvantage. For the null setting, exponential, Weibull,
and gamma survival distributions were considered. The proportional-hazards scenarios
covered hazard ratios of $1.5$, $2$, and $3$, whereas the two crossing scenarios were
generated from the Yang--Prentice model. The diffuse-transition scenario was based on a
piecewise-constant treatment hazard and was included to represent a more gradual temporal
change that cannot naturally be characterized by a single crossing point.

For each setting, we considered balanced sample sizes
\[
(n_1,n_2)\in\{(30,30),(50,50),(100,100),(200,200)\}
\]
and corresponding unbalanced designs
\[
(n_1,n_2)\in\{(30,60),(50,100),(100,200),(200,400)\}.
\]
Independent exponential censoring was chosen to obtain target event rates of approximately
$0.7$ and $0.9$. Each configuration was replicated $1000$ times. The bootstrap-based
procedures used $1000$ bootstrap samples per simulation replicate, and the critical value of
the max-$t$ statistic was approximated using $100{,}000$ Monte Carlo draws.

We compared the proposed procedures with the classical log-rank test, the
Gehan--Wilcoxon test, the adaptively weighted log-rank procedure by\cite{yangprentice09}, and the two-stage procedure by \cite{qiusheng08}.
Complete specifications of the data-generating mechanisms, including the Yang--Prentice
parameters, the diffuse-transition hazard, the censoring calibration, and the exact rejection
rates for all simulation settings, are provided in the Supplementary Material.

\subsection{Results}
Before comparing our methods with the competitors we wanted to determine if there are differences among our methods. %them. 
In principle four different hypothesis tests can be constructed using the stratified and pooled bootstrap and using the Wald and the max-t-type approach respectively. The rejection rates of these four under the null scenario can be seen in Figure~\ref{subfig:sim_null_our}. Here we pooled the results from the three distributions we examined. It becomes immediately obvious that the two tests based on the pooled bootstrap significantly outperform the ones based on the stratified bootstrap. This is not surprising, given the fact that the pooled bootstrap assumes the null hypothesis in its sampling scheme whereas the stratified one does not. Within each bootstrap procedure we can also see that the max-t-type approach performs a bit better than the Wald method. In order not to make the other figures too cluttered we therefore compare the competitors only to the method based on the pooled bootstrap and the max-t-type test.

In Figure~\ref{subfig:sim_null_all} we can see that our method keeps the nominal level well, performing similarly to most competitors, even outperforming some of them for small sample sizes. Only for sample size $200$ and event rate $0.7$ our method might be marginally more liberal, although confidence intervals show a large overlap.

Looking at Figure~\ref{subfig:sim_ph_all} we see the performance for the proportional hazards setting. Here it is not surprising that the log-rank test performs best, interestingly together with the adaptively weighted version. The other methods, including our method, are not very far behind but do show markedly lower power.

In the middle-crossing scenario (see Figure~\ref{subfig:sim_middle_all}), the weakness of the log-rank and Gehan-Wilcoxon test in picking up this alternative becomes apparent. The adaptively weighted log-rank test and especially the two-stage procedure perform particularly well here, with our method starting with power in between that of the other methods for small sample sizes and becoming competitive with the better methods at around a sample size of $100$.

In the end-crossing scenario (see Figure~\ref{subfig:sim_end_all}) the log-rank test performs worst, especially when the proportion of censored data is low. Our method performs similarly to the others, perhaps a bit worse than the Gehan-Wilcoxon and the adaptively weighted log-rank test for small samples but catching up for larger samples.

Under the diffuse transition alternative (Figure~\ref{subfig:sim_dta_all}), the two-stage procedure, which performed very well in all previous settings, is a bit worse than the others, together with the log-rank test. Here again our method is either in the middle, for smaller sample sizes, or one of the best two, at larger sample sizes.

Exact rejection rates for all scenarios can be found in the Tables~\ref{tab:null_balanced} through~\ref{tab:crossing_unbalanced}.

%\begin{figure}
%    \centering
%    \includegraphics[width=\linewidth]{graphics/Sim_Null_Our.pdf}
%    \caption{Simulation results of the null scenario (all three distributions combined). Shown are the rejection rates as well as $95\%$ Agresti Coull confidence intervals thereof. Here we restrict ourselves to the methods introduced in this paper, excluding the competitors. Points and errorbars have been jittered horizontally for clearer representation, but sample sizes are exactly as denoted by the ticks on the x-axis. Results are shown seperately for each combination of event rate (in the columns) and balancing (in the rows).}
%    \label{fig:placeholder}
%\end{figure}

%\begin{figure}
%    \centering
%    \includegraphics[width=\linewidth]{graphics/Sim_Null_All.pdf}
%    \caption{Simulation results of the null scenario (all three distributions combined). Shown are the rejection rates as well as $95\%$ Agresti Coull confidence intervals thereof. Here we show the competitors and our Max-t-type test for the pooled bootstrap. Points and error bars have been jittered horizontally for clearer representation, but sample sizes are exactly as denoted by the ticks on the x-axis. Results are shown seperately for each combination of event rate (in the columns) and balancing (in the rows).}
%    \label{fig:placeholder}
%\end{figure}

\begin{figure}
    \centering
    \begin{subfigure}[t]{0.85\linewidth}
        \centering
        \includegraphics[width=\linewidth]{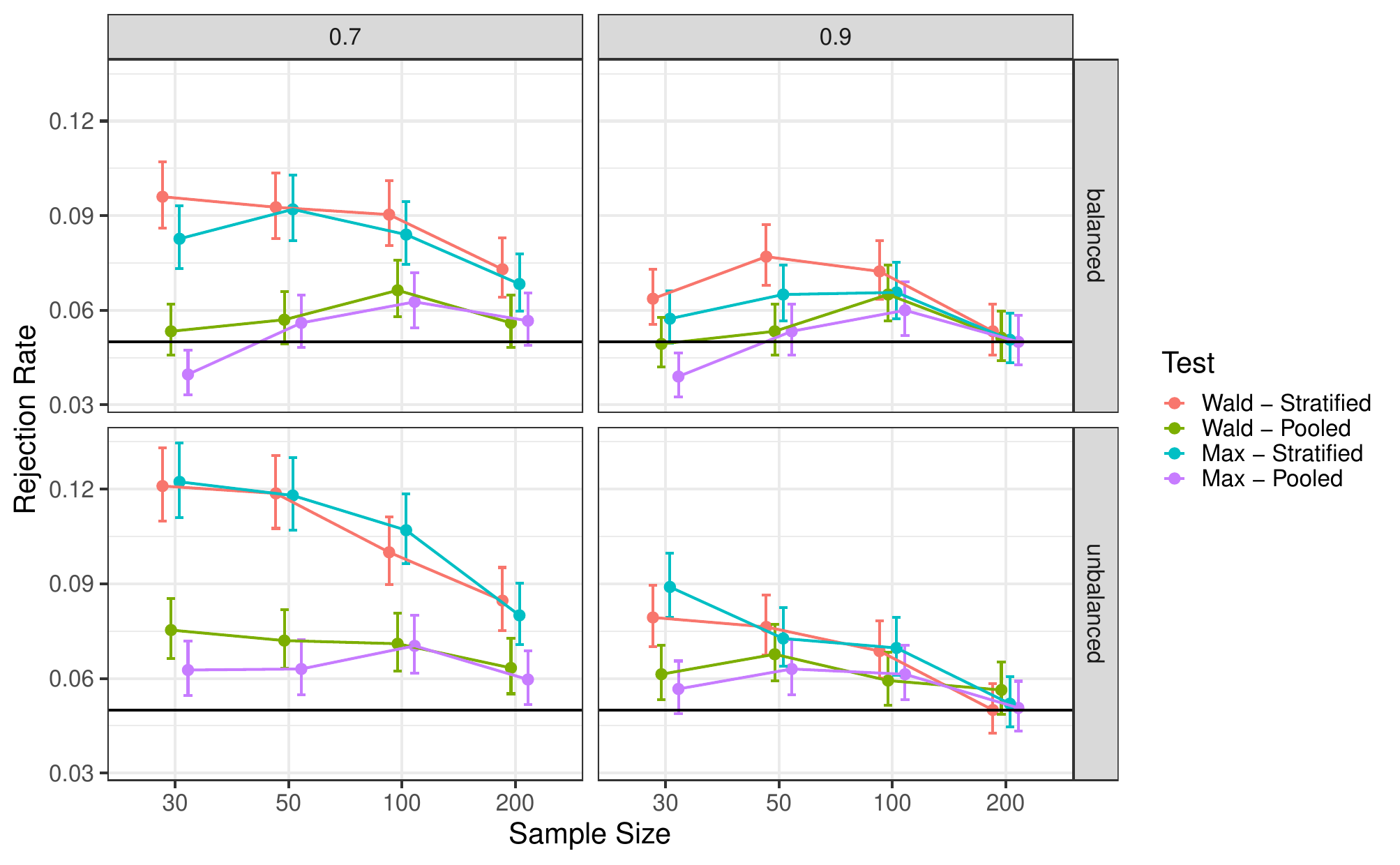}
        \caption{Comparison of the methods presented in this paper.\label{subfig:sim_null_our}}
    \end{subfigure}
    \hfill
    \begin{subfigure}[t]{0.85\linewidth}
        \centering
        \includegraphics[width=\linewidth]{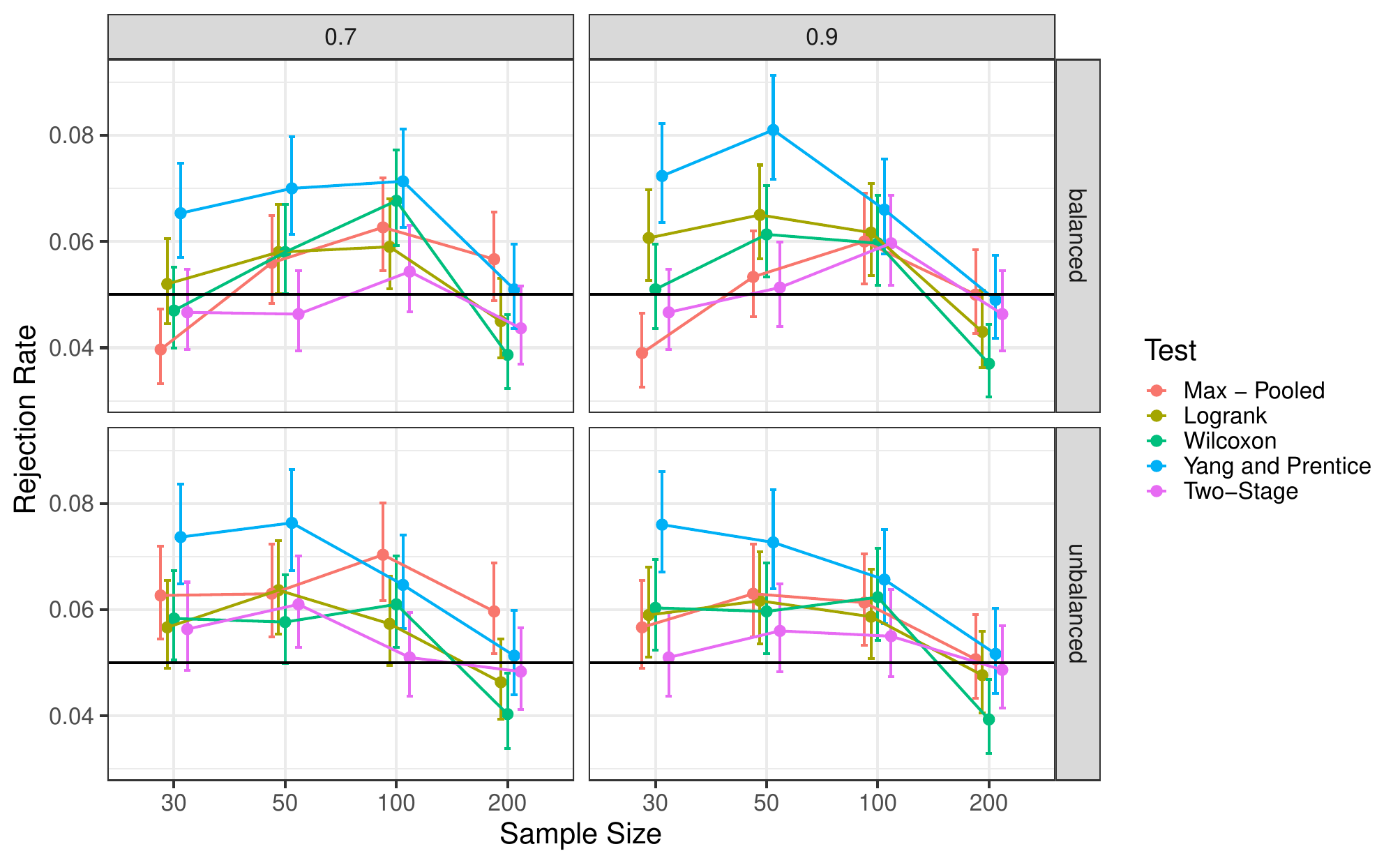}
        \caption{Comparison of the max-t-type test with the competitors.\label{subfig:sim_null_all}}
    \end{subfigure}
    \caption{Simulation results for the null scenario. Shown are the rejection rates as well as $95\%$ Agresti Coull confidence intervals thereof. Points and error bars have been jittered horizontally for clearer representation, but sample sizes are exactly as denoted by the ticks on the x-axis. Results are shown separately for each combination of event rate (in the columns) and balancing (in the rows).}
    \label{fig:sim_null}
\end{figure}

%\begin{figure}
%    \centering
%    \includegraphics[width=\linewidth]{graphics/Sim_PH_All.pdf}
%    \caption{Simulation results of the proportional hazards scenario. The hazard rate is indicated by linetype. Shown are the rejection rates as well as $95\%$ Agresti Coull confidence intervals thereof. Here we show the competitors and our Max-t-type test for the pooled bootstrap. Points and error bars have been jittered horizontally for clearer representation, but sample sizes are exactly as denoted by the ticks on the x-axis. Results are shown seperately for each combination of event rate (in the columns) and balancing (in the rows).}
%    \label{fig:placeholder}
%\end{figure}

%\begin{figure}
%    \centering
%    \includegraphics[width=\linewidth]{graphics/Sim_Middle_All.pdf}
%    \caption{Simulation results of the crossing-in-the-middle scenario. Shown are the rejection rates as well as $95\%$ Agresti Coull confidence intervals thereof. Here we show the competitors and our Max-t-type test for the pooled bootstrap. Points and error bars have been jittered horizontally for clearer representation, but sample sizes are exactly as denoted by the ticks on the x-axis. Results are shown seperately for each combination of event rate (in the columns) and balancing (in the rows).}
%    \label{fig:placeholder}
%\end{figure}

\begin{figure}
    \centering
    \begin{subfigure}[t]{0.8\linewidth}
        \centering
        \includegraphics[width=\linewidth]{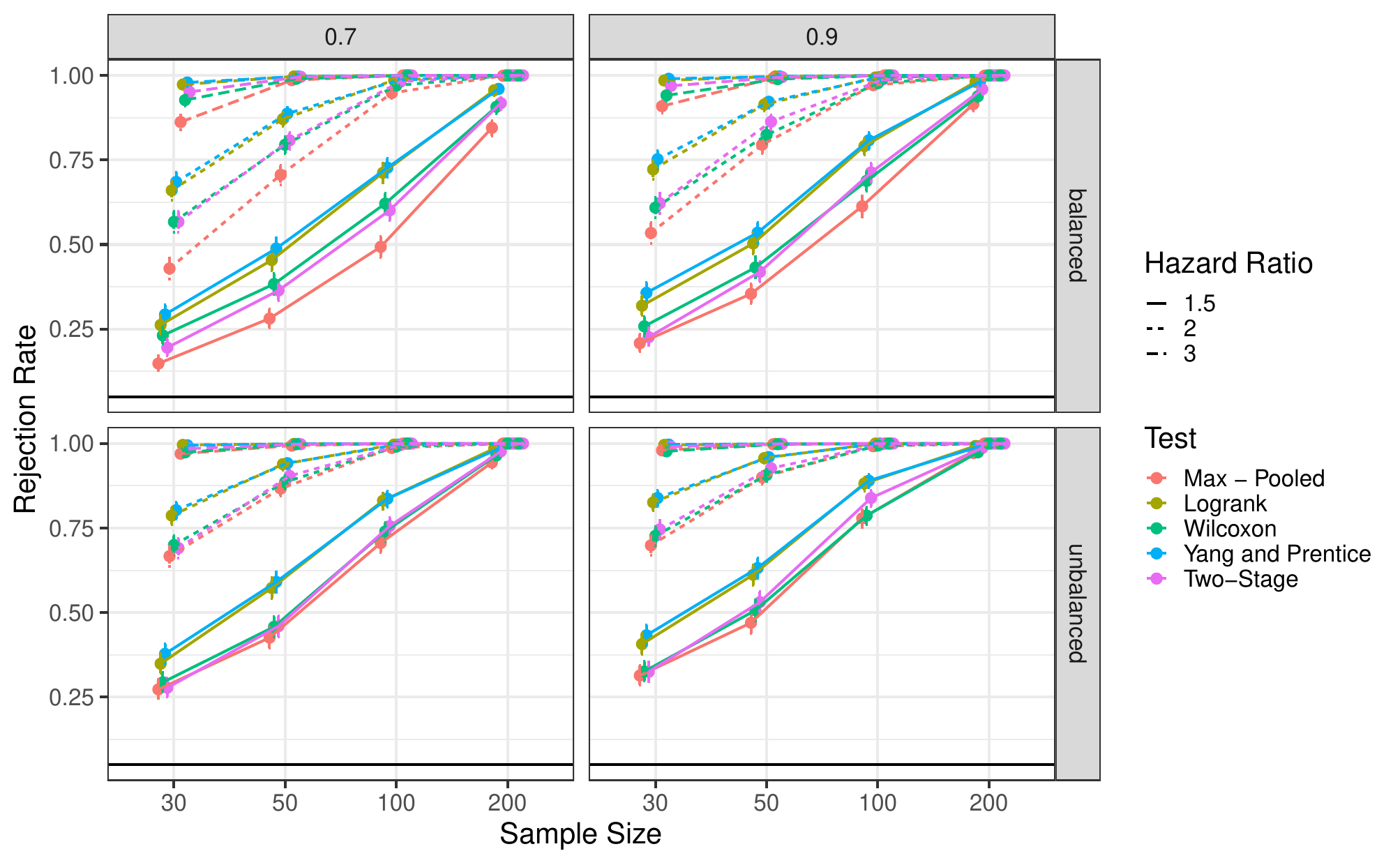}
        \caption{Results for the proportional hazards scenario. The hazard ratio is indicated by linetype.\label{subfig:sim_ph_all}}
    \end{subfigure}
    \hfill
    \begin{subfigure}[t]{0.8\linewidth}
        \centering
        \includegraphics[width=\linewidth]{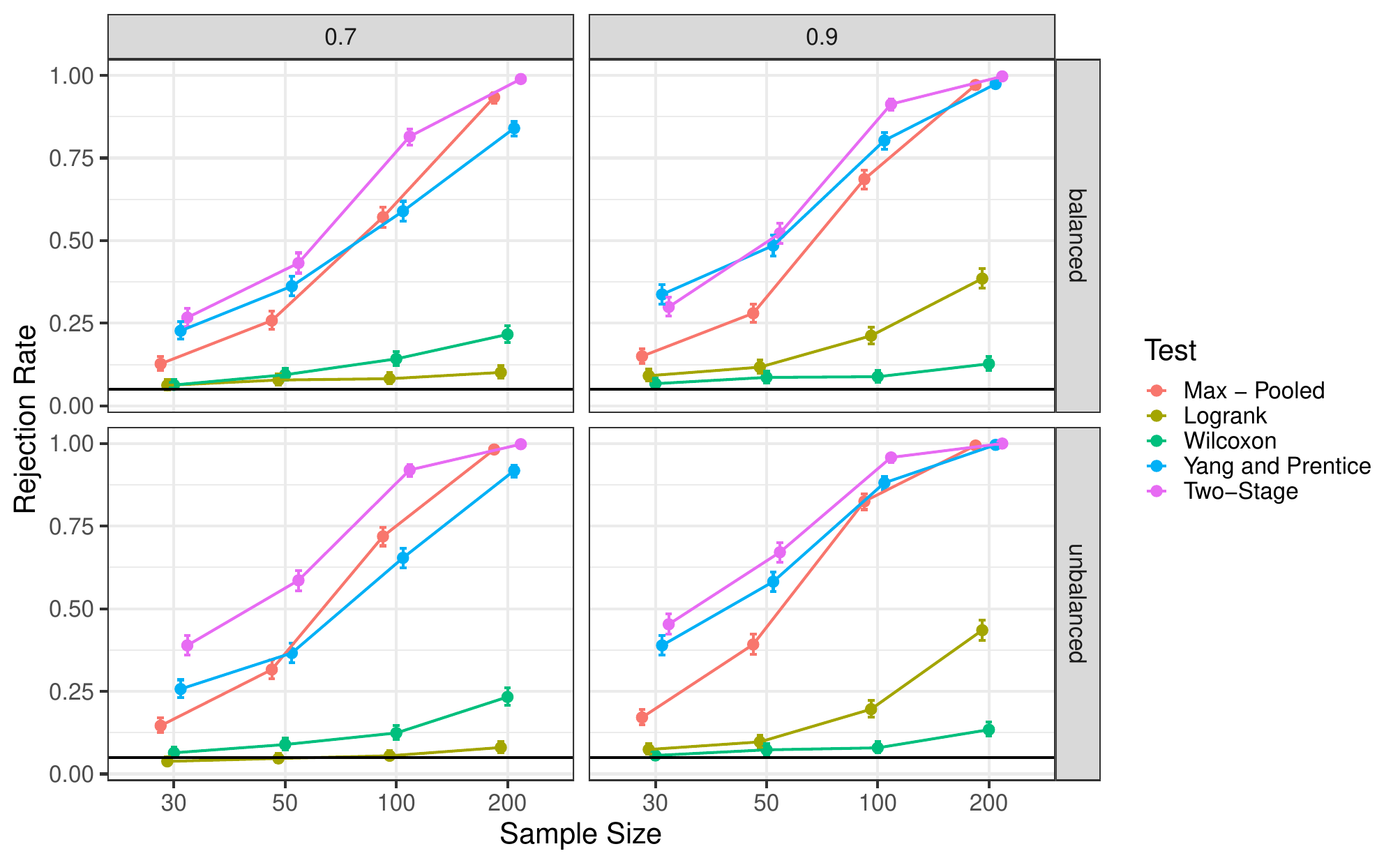}
        \caption{Results for the middle-crossing scenario.\label{subfig:sim_middle_all}}
    \end{subfigure}
    \caption{Simulation results for the proportional hazards and middle-crossing scenarios. Here we compare the max-t-type method based on the pooled bootstrap with the competitors. Shown are the rejection rates as well as $95\%$ Agresti Coull confidence intervals thereof. Points and error bars have been jittered horizontally for clearer representation, but sample sizes are exactly as denoted by the ticks on the x-axis. Results are shown separately for each combination of event rate (in the columns) and balancing (in the rows).}
    \label{fig:sim_ph_middle}
\end{figure}

%\begin{figure}
%    \centering
%    \includegraphics[width=\linewidth]{graphics/Sim_End_All.pdf}
%    \caption{Simulation results of the crossing-at-the-end scenario. Shown are the rejection rates as well as $95\%$ Agresti Coull confidence intervals thereof. Here we show the competitors and our Max-t-type test for the pooled bootstrap. Points and error bars have been jittered horizontally for clearer representation, but sample sizes are exactly as denoted by the ticks on the x-axis. Results are shown seperately for each combination of event rate (in the columns) and balancing (in the rows).}
%    \label{fig:placeholder}
%\end{figure}

%\begin{figure}
%    \centering
%    \includegraphics[width=\linewidth]{graphics/Sim_DTA_All.pdf}
%    \caption{Simulation results of the diffuse transition alternative scenario. Shown are the rejection rates as well as $95\%$ Agresti Coull confidence intervals thereof. Here we show the competitors and our Max-t-type test for the pooled bootstrap. Points and error bars have been jittered horizontally for clearer representation, but sample sizes are exactly as denoted by the ticks on the x-axis. Results are shown seperately for each combination of event rate (in the columns) and balancing (in the rows).}
%    \label{fig:placeholder}
%\end{figure}

\begin{figure}
    \centering
    \begin{subfigure}[t]{0.8\linewidth}
        \centering
        \includegraphics[width=\linewidth]{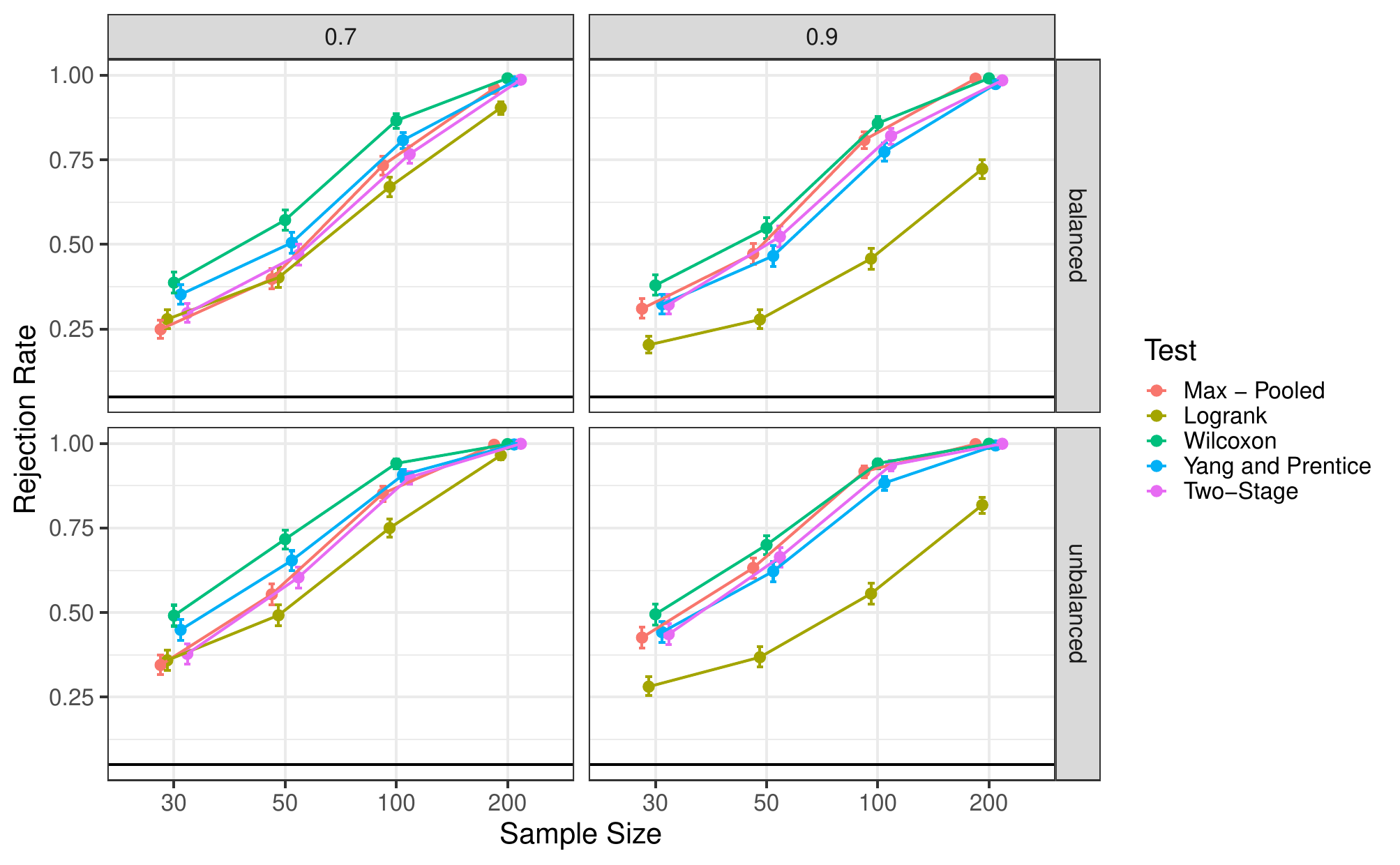}
        \caption{Results for the end-crossing scenario.\label{subfig:sim_end_all}}
    \end{subfigure}
    \hfill
    \begin{subfigure}[t]{0.8\linewidth}
        \centering
        \includegraphics[width=\linewidth]{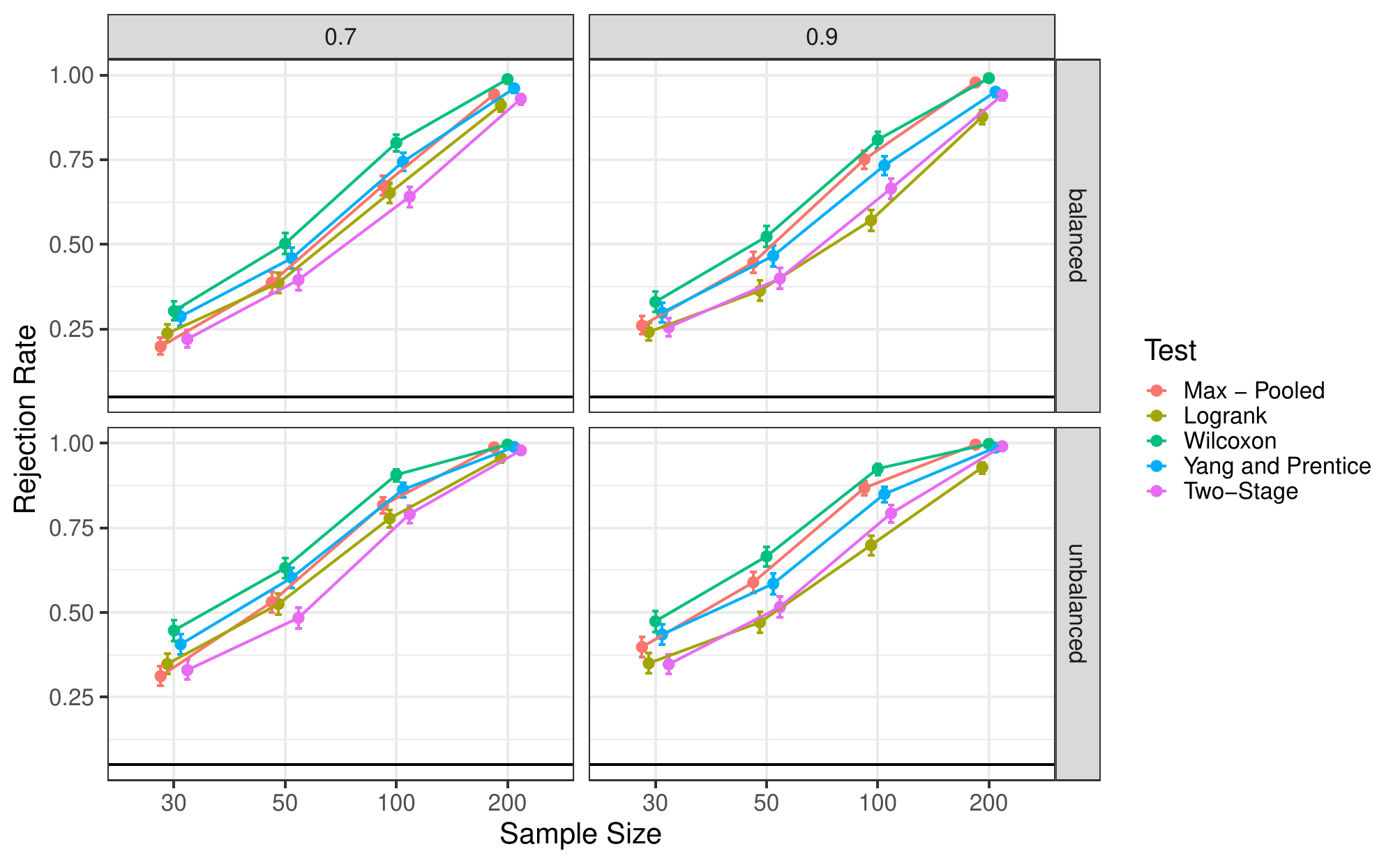}
        \caption{Results for the diffuse transition alternative scenario.\label{subfig:sim_dta_all}}
    \end{subfigure}
    \caption{Simulation results for the end-crossing and diffuse transition alternative scenarios. Here we compare the max-t-type method based on the pooled bootstrap with the competitors. Shown are the rejection rates as well as $95\%$ Agresti Coull confidence intervals thereof. Points and error bars have been jittered horizontally for clearer representation, but sample sizes are exactly as denoted by the ticks on the x-axis. Results are shown separately for each combination of event rate (in the columns) and balancing (in the rows).}
    \label{fig:sim_end_dta}
\end{figure}

\section{Data Example: Melanoma
Brain Metastases}
We set
$ F = F_{\mathrm{Obs}},$
    and
   $ G = F_{\mathrm{WBRT}},$
so that \(F\) denotes the observation arm and \(G\) the WBRT arm. Since death is
an adverse event, values \(\theta(F,G)>1/2\) indicate that observation patients
tend to die earlier than WBRT patients, corresponding to a survival advantage
of WBRT on the global pairwise probability scale.

The estimated joint effect was
\[
    \widehat\psi
    =
    (\widehat\theta,\widehat O)^\top
    =
    (0.522,\;0.401)^\top .
\]
The estimate \(\widehat\theta=0.522\) indicates only a small overall pairwise
advantage for WBRT. The simultaneous max-t-type confidence interval was
\[
    0.432 \leq \theta(F,G) \leq 0.613,
\]
and therefore includes the null value \(1/2\). This agrees with the published
overall-survival analysis, which found no significant global survival difference.

The temporal component provides a more detailed description of the pattern seen
in the Kaplan--Meier curves. The estimate
\(
    \widehat O = 0.401
\)
is below its null value \(1/2\), suggesting that the pairwise survival advantage
of WBRT is stronger in the early part of follow-up than in the late part. The
corresponding max-t-type confidence interval was
\[
    0.291 \leq O(F,G) \leq 0.510.
\]
Although this interval still includes \(1/2\), its upper endpoint is close to
the null value, indicating that the observed deviation is mainly driven by the
temporal contrast.

For interpretation, we decompose the estimate into early and late conditional
effects:
\[
    \widehat\theta_{\le}
    =
    \widehat\theta-\frac12\widehat O
    =
    0.322,
    \qquad
    \widehat\theta_{>}
    =
    \widehat\theta+\frac12\widehat O
    =
    0.723.
\]
Under equality of the two survival distributions, the corresponding reference
values are \(1/4\) and \(3/4\). Hence the estimated early and late deviations are
\[
    \widehat\delta_E
    =
    \widehat\theta_{\le}-\frac14
    =
    0.072,
    \qquad
    \widehat\delta_L
    =
    \widehat\theta_{>}-\frac34
    =
    -0.027.
\]
Thus, the point estimates are consistent with an early survival advantage of
WBRT followed by attenuation or reversal later in follow-up. This agrees with
the visual impression from Figure~\ref{fig:km_wbrt}. However, because the
confidence intervals for both components contain their respective null values,
the example should be interpreted descriptively. It illustrates how the proposed
bivariate estimand can quantify an early--late survival pattern that is not
captured by a single proportional-hazards summary or by a standard global
log-rank comparison.

\begin{figure}
    \centering

    \begin{subfigure}[t]{0.58\linewidth}
        \centering
        \includegraphics[width=\linewidth]{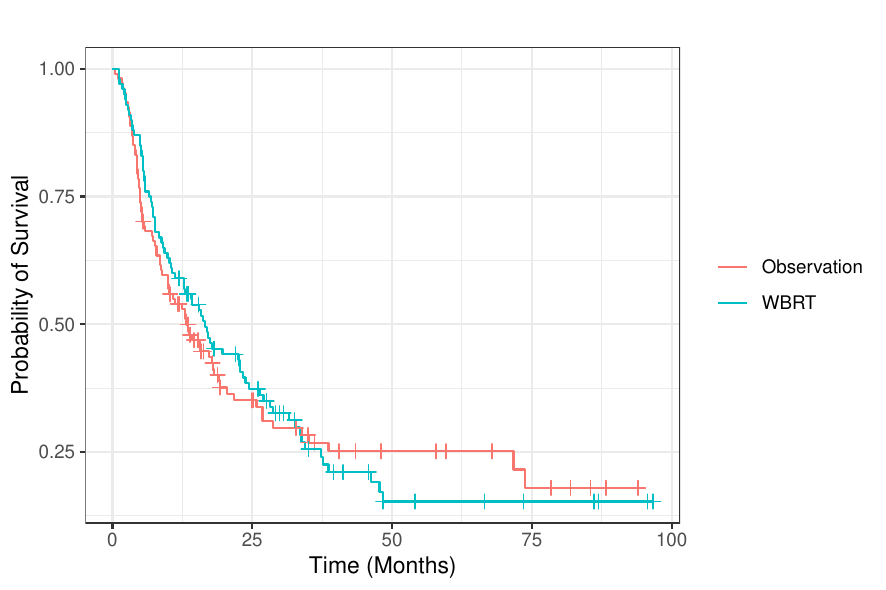}
        \caption{Kaplan--Meier estimates of overall survival.}
        \label{fig:km_wbrt}
    \end{subfigure}
    \hfill
    \begin{subfigure}[t]{0.38\linewidth}
        \centering
        \includegraphics[width=\linewidth]{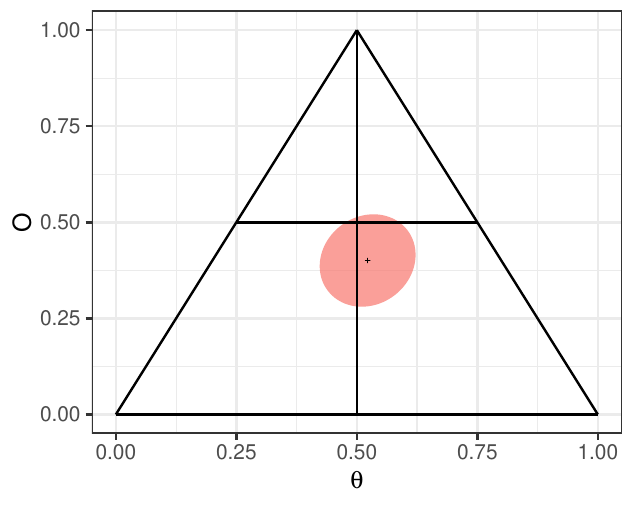}
        \caption{Estimated joint effect with bootstrap uncertainty region.}
        \label{fig:wbrt-ellipse}
    \end{subfigure}

    \caption{Overall survival analysis in the WBRT example.
    \textbf{(a)} Kaplan--Meier estimates for patients assigned to observation or
    adjuvant whole-brain radiation therapy (WBRT) after local treatment of
    melanoma brain metastases. The curves suggest a non-proportional pattern,
    with a small early advantage for WBRT that attenuates and appears to reverse
    later in follow-up.
    \textbf{(b)} Estimated joint effect
    \((\widehat\theta,\widehat O)=(0.522,0.401)\). The ellipse shows the
    bootstrap-based joint uncertainty region, and the null value
    \((1/2,1/2)^\top\) is shown for reference.}
    \label{fig:wbrt-example2}
\end{figure}

An additional real data example illustrating the use of our method is provided in the Appendix~\ref{app:additional_example}.

\section{Discussion}

We proposed a two-dimensional nonparametric framework for comparing right-censored survival distributions under non-proportional hazards. The method combines a Kaplan-Meier-based Mann-Whitney effect, measuring the global pairwise ordering of event times, with a median-split temporal contrast, describing whether this ordering differs between earlier and later follow-up. Thus, the approach provides an interpretable probability-scale summary of both overall and temporal survival differences.

The proposed method is particularly useful when early and late effects point in different directions, so that a single global summary or a log-rank test may obscure relevant structure. In this sense, the procedure is targeted rather than omnibus: it is designed to quantify clinically interpretable early-late patterns, not to detect every possible difference between survival distributions.

The data example illustrates this interpretation. While the global effect suggested only a small overall advantage of WBRT, the temporal component was compatible with an early advantage followed by attenuation or reversal later in follow-up. Since the simultaneous confidence intervals included the null values, this finding should be viewed as descriptive rather than confirmatory.

Several extensions appear natural. Future work could develop covariate-adjusted versions of the proposed effects, and extend the framework to competing risks.

\subsubsection*{Funding}
The first author gratefully acknowledges the support of the WISS 2025 project `IDA-lab Salzburg' (20204-WISS/225/197-2019542 and 20102-F1901166-KZP).

\clearpage
\thispagestyle{plain}
\null
\vfill
\begin{center}
{\LARGE\bfseries Supplementary Material for\\[0.5em]
\textit{Detecting Early and Late Divergences in Survival Curves Using
Nonparametric Effect Measures}\par}

\vspace{2em}
{\large Patrick B. Langthaler \quad Jun Ma \quad Jonas Beck\par}
\end{center}
\vfill
\clearpage

\appendix

% Restart all counters used in the supplement.
\setcounter{section}{0}
\setcounter{subsection}{0}
\setcounter{subsubsection}{0}
\setcounter{equation}{0}
\setcounter{figure}{0}
\setcounter{table}{0}
\setcounter{theorem}{0}
\setcounter{lemma}{0}
\setcounter{definition}{0}
\setcounter{remark}{0}
\setcounter{proposition}{0}
\setcounter{assumption}{0}

% Supplement numbering.
\renewcommand{\thesection}{S\arabic{section}}
\renewcommand{\thesubsection}{\thesection.\arabic{subsection}}
\renewcommand{\thesubsubsection}{\thesubsection.\arabic{subsubsection}}
\renewcommand{\theequation}{S\arabic{equation}}
\renewcommand{\thefigure}{S\arabic{figure}}
\renewcommand{\thetable}{S\arabic{table}}
\renewcommand{\thetheorem}{S\arabic{theorem}}
\renewcommand{\thelemma}{S\arabic{lemma}}
\renewcommand{\thedefinition}{S\arabic{definition}}
\renewcommand{\theremark}{S\arabic{remark}}
\renewcommand{\theproposition}{S\arabic{proposition}}
\renewcommand{\theassumption}{S\arabic{assumption}}

% Prefix hyperlink destinations because the visible counters restart at one.
\providecommand*{\theHsection}{}
\providecommand*{\theHsubsection}{}
\providecommand*{\theHsubsubsection}{}
\providecommand*{\theHequation}{}
\providecommand*{\theHfigure}{}
\providecommand*{\theHtable}{}
\providecommand*{\theHtheorem}{}
\providecommand*{\theHlemma}{}
\providecommand*{\theHdefinition}{}
\providecommand*{\theHremark}{}
\providecommand*{\theHproposition}{}
\providecommand*{\theHassumption}{}
\renewcommand*{\theHsection}{supp.\arabic{section}}
\renewcommand*{\theHsubsection}{supp.\arabic{section}.\arabic{subsection}}
\renewcommand*{\theHsubsubsection}{supp.\arabic{section}.\arabic{subsection}.\arabic{subsubsection}}
\renewcommand*{\theHequation}{supp.\arabic{equation}}
\renewcommand*{\theHfigure}{supp.\arabic{figure}}
\renewcommand*{\theHtable}{supp.\arabic{table}}
\renewcommand*{\theHtheorem}{supp.\arabic{theorem}}
\renewcommand*{\theHlemma}{supp.\arabic{lemma}}
\renewcommand*{\theHdefinition}{supp.\arabic{definition}}
\renewcommand*{\theHremark}{supp.\arabic{remark}}
\renewcommand*{\theHproposition}{supp.\arabic{proposition}}
\renewcommand*{\theHassumption}{supp.\arabic{assumption}}

\section{Proofs}
Let
\[
\tau_0
:=
\sup\left\{
t>0:
\Pr(X\geq t)>0
\text{ and }
\Pr(Y\geq t)>0
\right\}
\in(0,\infty].
\]
For the local asymptotic arguments, fix an arbitrary finite
$\tau<\tau_0$ such that the median $m$ of $G$ satisfies $m<\tau$.
Assume that the Kaplan--Meier estimators admit weak limits on
$[0,\tau]$ under the supremum norm.
In this proof we incorporate similar techniques as in \cite{Dobler2018}:

Let $D[0,\tau ]$ be the space of càdlàg functions on $[0,\tau]$ equipped with $\|\cdot\|_\infty$.
We write $\rightsquigarrow$ for weak
convergence.
For bootstrap random elements, we write
\(
X_N^* \rightsquigarrow_{\mathbb P} X
\)
for conditional weak convergence in outer probability, that is, the
conditional law of $X_N^*$ given the data converges weakly in outer
probability to the law of $X$.
For a càdlàg cdf $A$, write $A(t-)$ for left limits and $A^{\pm}(t):=(A(t-)+A(t))/2$.

Assume standard regularity conditions ensuring the KM functional CLT on $[0,\tau]$:
there exist tight mean-zero Gaussian processes $\mathbb G_F,\mathbb G_G$ in $D[0,\tau]$ such that
\begin{equation} \label{KM_CLT}
\sqrt{n_1}\,(\widehat F-F)\rightsquigarrow \mathbb G_F,\qquad
\sqrt{n_2}\,(\widehat G-G)\rightsquigarrow \mathbb G_G,
\end{equation}
and $\mathbb G_F$ and $\mathbb G_G$ are independent (due to group independence). This follows directly from Example 3.10.33 in \cite{vdVW2023}.
With $N=n_1+n_2$ and $n_1/N\to\lambda\in(0,1)$,
\[
\sqrt N\binom{\widehat F-F}{\widehat G-G}\rightsquigarrow
\binom{\lambda^{-1/2}\mathbb G_F}{(1-\lambda)^{-1/2}\mathbb G_G}
=: \mathbb Z
\quad\text{in }D[0,\tau]^2.
\]

For the stratified (within-group) nonparametric bootstrap, assume the conditional KM bootstrap CLT:
conditionally on the data,
\[
\sqrt{n_1}(\widehat F^*-\widehat F)
\rightsquigarrow_{\mathbb P}\mathbb G_F,
\qquad
\sqrt{n_2}(\widehat G^*-\widehat G)
\rightsquigarrow_{\mathbb P}\mathbb G_G,
\]
in $D[0,\tau]$, in outer probability, with conditional independence across groups.
Equivalently,
\begin{equation} \label{KM_conditional_CLT}
\sqrt N
\begin{pmatrix}
\widehat F^*-\widehat F\\
\widehat G^*-\widehat G
\end{pmatrix}
\rightsquigarrow_{\mathbb P}\mathbb Z.
\end{equation}

For the pooled bootstrap, we have
\begin{align*}
\sqrt{n_1} \left( \widehat{H}_1 - H \right) \rightsquigarrow \mathbb{G}_H
\qquad
\sqrt{n_2} \left( \widehat{H}_2 - H \right) \rightsquigarrow \mathbb{G}'_H
\end{align*}
where $\mathbb{G}_H$ and $\mathbb{G}'_H$ are two independent versions of the same tight mean zero Gaussian process. Due to independence of $\widehat{H}_1$ and $\widehat{H}_2$ we get
\begin{equation*}
\sqrt{N} \left( (\widehat{H}_1, \widehat{H}_2) - (H, H) \right) \rightsquigarrow (\lambda^{-1/2} \mathbb{G}_H, (1 - \lambda)^{-1/2} \mathbb{G}'_H) =: \mathbb{Z}_0
\end{equation*}
The two-sample bootstrap result from Theorem 3.8.6 of~\cite{vdVW2023} gives
\begin{equation*}
\sqrt{n_1} \left( \widehat{H}^{*}_1 - \widehat{H}_1 \right) \rightsquigarrow_{\mathbb P} \mathbb{G}_H
\qquad
\sqrt{n_2} \left( \widehat{H}^{*}_2 - \widehat{H}_2 \right) \rightsquigarrow_{\mathbb P} \mathbb{G}'_H
\end{equation*}
Conditional on the data we have independence of $(\widehat{H}^{*}_1, \widehat{H}_1)$ and $(\widehat{H}^{*}_2, \widehat{H}_2)$ and thus
\begin{equation*}
\sqrt{N} \left( (\widehat{H}^{*}_1, \widehat{H}^{*}_2) - (\widehat{H}_1, \widehat{H}_2) \right) \rightsquigarrow_{\mathbb P} \mathbb{Z}_0
\end{equation*}

Whenever $a$ has bounded variation on $[0,\tau]$ and $h\in \ell^\infty([0,\tau])$, define
\[
\int_{[0,\tau]} a\,dh := a(t)h(t)\big|_{0}^{\tau-} - \int_{[0,\tau]} h\,da,
\]
so that $h\mapsto \int a\,dh$ is a continuous linear functional on $\ell^\infty([0,\tau])$.
(Here the boundary term is interpreted as $\lim_{t\uparrow \tau} a(t)h(t)$.)
Let
\[
\mathbb H :=
\Bigl\{(h_F,h_G)\in D[0,\tau]^2:\ h_F^{\pm}\in L^1_{|[0,\tau]}(G),\ h_G^{\pm}\in L^1_{|[0,\tau]}(F),
\ h_F(0)=h_G(0)=0 \Bigr\}.
\]
(All relevant integrals below are well-defined for directions in $\mathbb H$ because $F$ and $G$
are monotone and hence of bounded variation.)

\begin{assumption}[Regularity conditions]
\label{ass:regularity}
Let
\[
\tau_0
:=
\sup\left\{
t>0:
\Pr(X\geq t)>0
\text{ and }
\Pr(Y\geq t)>0
\right\}
\in(0,\infty].
\]
Suppose that the following conditions hold.

\begin{enumerate}[label=(\roman*)]
\item Within each group, the observations are independent and identically
distributed, the two groups are mutually independent, and censoring is
independent of the event time:
\[
T_i\perp C_i,
\qquad
U_j\perp D_j.
\]

\item
With \(N=n_1+n_2\),
\[
\frac{n_1}{N}\longrightarrow\lambda\in(0,1).
\]

\item
The median
\[
m:=G^{-1}(1/2)
\]
is unique and satisfies \(m<\tau_0\). Moreover, \(G\) is continuously
differentiable in a neighborhood of \(m\), with density \(g\) satisfying
\[
g(m)>0.
\]
In addition, \(F\) is continuous at \(m\).

\item
For every finite \(\tau\) satisfying \(m<\tau<\tau_0\), the
Kaplan--Meier estimators satisfy
\[
\sqrt{n_1}(\widehat F-F)\rightsquigarrow \mathbb G_F,
\qquad
\sqrt{n_2}(\widehat G-G)\rightsquigarrow \mathbb G_G
\]
in \(D[0,\tau]\) equipped with the supremum norm, where
\(\mathbb G_F\) and \(\mathbb G_G\) are independent tight mean-zero
Gaussian processes whose sample paths belong almost surely to the
tangent space of the effect functional.

\item \label{ass:tail-localization}
For every \(\varepsilon>0\),
\[
\lim_{\tau\uparrow\tau_0}
\limsup_{N\to\infty}
\Pr\left[
\sqrt N
\left\|
\{\widehat\psi_N-\psi(F,G)\}
-
\{\widehat\psi_{N,\tau}-\psi_\tau(F,G)\}
\right\|>\varepsilon
\right]
=0.
\]

\item
The covariance matrices of the localized Gaussian limits satisfy
\[
\Sigma_\tau(F,G)\longrightarrow\Sigma(F,G)
\qquad\text{as }\tau\uparrow\tau_0,
\]
where \(\Sigma(F,G)\) is finite.
\end{enumerate}
\end{assumption}
\begin{lemma} \label{lemma1}
   Fix $\tau>0$. Define, for distribution functions $F,G$ on $[0,\tau]$,
\[
\Phi(F,G):=\int_{[0,\tau]} F^{\pm}(t)\,dG(t),
\qquad F^{\pm}(t):=\frac{F(t-)+F(t)}{2}.
\]
Then $\Phi$ is Hadamard differentiable at $(F,G)$ tangentially to
$D[0,\tau]\times C[0,\tau]$ %(or any subset where the perturbations remain distribution functions)
, with derivative
\[
\dot\Phi_{(F,G)}(h_F,h_G)
=
\int_{[0,\tau]} h_F^{\pm}(t)\,dG(t)
\;+\;
\int_{[0,\tau]} F^{\pm}(t)\,dh_G(t),
\]
where $\int F^{\pm}\,dh_G$ is understood via integration by parts
$\int F^{\pm}\,dh_G
= F^{\pm}(\tau)h_G(\tau)-F^{\pm}(0)h_G(0)-\int h_G^{-}\,dF^{\pm}$.
\end{lemma}
\begin{proof}
    This is a direct specialization of Lemma~3.10.18 in \cite{vdVW2023} to the map $(A,B)\mapsto \int A\,dB$ on $[0,\tau]$:
take $A=F^{\pm}$ and $B=G$ (note that $G$ is of bounded variation with total
variation $1$). The derivative in Lemma~3.10.18 is
$\int A\,d\beta + \int \alpha\,dB$ for perturbations $(\alpha,\beta)$, which becomes
$\int F^{\pm}\,dh_G + \int h_F^{\pm}\,dG$.
\end{proof}
For fixed $\tau<\tau_0$, define
\[
\Theta_\tau(F,G)
:=
\int_{[0,\tau]}F^\pm(t)\,\mathrm{d}G(t)
+
\frac{1+F(\tau)}{2}\{1-G(\tau)\}.
\]
\begin{lemma} \label{lemma2}
For every fixed $\tau<\tau_0$, the map $\Theta_\tau$ is Hadamard differentiable at $(F,G)$ tangentially to
$D[0,\tau]\times C[0,\tau]$ with derivative
\[
\dot\Theta_{\tau,(F,G)}(h_F,h_G)
=
\int_{[0,\tau]} h_F(t)\,dG(t)
+
\int_{[0,\tau]} F(t)\,dh_G(t)
+\frac{1}{2}\,(1-G(\tau))\,h_F(\tau)
-\frac{1+F(\tau)}{2}\,h_G(\tau).
\]
\end{lemma}
\begin{proof}
    The first two terms are the derivative of $(F,G)\mapsto \int_{[0,\tau]}F\,dG$ by Lemma \ref{lemma1}. The correction term depends only on the 2-vector $(F(\tau),G(\tau))$ via the smooth
finite-dimensional map
\[
g(x,y)=\frac{1+x}{2}(1-y).
\]
Hence it is Hadamard differentiable with gradient
$\nabla g(x,y)=\bigl(\frac12(1-y),-\frac12(1+x)\bigr)$, and by the chain rule (Lemma~3.10.3 in \cite{vdVW2023}) its contribution is
$\frac12(1-G(\tau))h_F(\tau)-\frac{1+F(\tau)}{2}h_G(\tau)$.
\end{proof}

\begin{lemma} \label{lemma3}
    Let $0<p<1$ and let $G$ be continuously differentiable on a neighborhood of
$m:=G^{-1}(p)$ with derivative $g(m)>0$.
Then the inverse map $Q(G):=G^{-1}(p)$ is Hadamard differentiable at $G$
tangentially to $C[0,\tau]$, with derivative
\[
\dot Q_G(h_G)= -\frac{h_G(m)}{g(m)}.
\]
In particular, for the median $p=1/2$ this yields the derivative used for the split point.
\end{lemma}
\begin{proof}
    Directly by Lemma~3.10.24 in \cite{vdVW2023}.
\end{proof}

\begin{lemma} \label{lemma4}
    Let $m:=G^{-1}(1/2)$ be unique, and assume $G$ is continuous at $m$ and
continuously differentiable in a neighborhood of $m$ with $g(m)>0$.
Define
\[
A(F,G):=\int_{[0,m]} F^{\pm}(t)\,dG(t).
\]
Then $A$ is Hadamard differentiable at $(F,G)$ tangentially to
$D[0,\tau]\times C[0,\tau]$, with derivative
\[
\dot A_{(F,G)}(h_F,h_G)
=
\int_{[0,m]} h_F^{\pm}(t)\,dG(t)
\;+\;
\int_{[0,m]} F^{\pm}(t)\,dh_G(t)
\;-\;
F^{\pm}(m)\,h_G(m).
\]
\end{lemma}

\begin{proof}
Let $t_n\downarrow 0$, $(h_{F,n},h_{G,n})\to(h_F,h_G)$ uniformly, and set
$F_n:=F+t_n h_{F,n}$, $G_n:=G+t_n h_{G,n}$, assuming $(F_n,G_n)$ remain
distribution functions on $[0,\tau]$.
Let $m_n:=G_n^{-1}(1/2)$. By Lemma \ref{lemma3}, $m_n\to m$ and
\[
m_n = m - t_n\frac{h_G(m)}{g(m)} + o(t_n).
\]

Decompose
\[
A(F_n,G_n)-A(F,G)
=
\Bigl[\int_{[0,m]} F_n^{\pm}\,dG_n - \int_{[0,m]} F^{\pm}\,dG\Bigr]
+
\Bigl[\int_{(m,m_n]} F_n^{\pm}\,dG_n\Bigr],
\]
with the second bracket interpreted as $-\int_{(m_n,m]}F_n^{\pm}\,dG_n$ if $m_n<m$.

For the first bracket, apply Lemma~\ref{lemma1} on the fixed interval $[0,m]$:
\[
\int_{[0,m]} F_n^{\pm}\,dG_n - \int_{[0,m]} F^{\pm}\,dG
=
t_n\Bigl(
\int_{[0,m]} h_F^{\pm}\,dG + \int_{[0,m]} F^{\pm}\,dh_G
\Bigr) + o(t_n).
\]

It remains to linearize the moving-boundary term
$R_n:=\int_{(m,m_n]}F_n^{\pm}\,dG_n$.
Fix $\varepsilon>0$. By the continuity of $F^{\pm}$ in $m$, choose $\delta>0$ such that
$|F^{\pm}(t)-F^{\pm}(m)|\le \varepsilon$ for $t\in[m-\delta,m+\delta]$. 
For $n$ large, $m_n\in[m-\delta,m+\delta]$, and
\[
R_n
=
F^{\pm}(m)\bigl(G_n(m_n)-G_n(m)\bigr)
+
r_n,
\qquad |r_n|\le \varepsilon\,|G_n(m_n)-G_n(m)|.
\]

Since $G$ is continuous at $m$, $G(m)=1/2$.
By definition of $m_n$, $G_n(m_n)=1/2$, hence
\[
G_n(m_n)-G_n(m)=\frac12 - \bigl(G(m)+t_n h_{G,n}(m)\bigr)
= -t_n h_{G,n}(m),
\]
and therefore
\[
R_n
=
-t_n F^{\pm}(m)\,h_{G, n}(m) + o(t_n).
\]
Since \(h_{G,n}(m) \to h_G(m)\), the sequence
\(\{h_{G,n}(m)\}_{n \in \mathbb{N}}\) is bounded. Hence, for every
\(\varepsilon > 0\),
\[
\limsup_{n \to \infty}
\frac{|r_n|}{t_n}
\leq
\varepsilon
\limsup_{n \to \infty}
|h_{G,n}(m)|.
\]
Since \(\varepsilon > 0\) is arbitrary, it follows that
\(r_n = o(t_n)\). Together with \(h_{G,n}(m) \to h_G(m)\), this yields
\[
R_n
=
-t_n F^{\pm}(m) h_G(m) + o(t_n),
\qquad n \to \infty.
\]
because $h_{G,n}(m)\to h_G(m)$ and $|r_n|/t_n\le \varepsilon\,|h_{G,n}(m)|$.
Letting $\varepsilon\downarrow 0$ gives the stated boundary contribution
$-F^{\pm}(m)h_G(m)$.

Combining the two parts yields the derivative formula.
\end{proof}

\begin{lemma} \label{lemma5}
    Under the assumptions of Lemma \ref{lemma4},
let
%\[
%\Theta(F,G):=\int_{[0,\tau]} F^{\pm}(t)\,dG(t),
%\qquad
%A(F,G):=\int_{[0,m]} F^{\pm}(t)\,dG(t).
%\]
%Then
\[
O_\tau(F,G)
:=
2\Theta_\tau(F,G)-4A(F,G).
\]
Then $O_\tau$ is Hadamard differentiable with derivative
\[
\dot O_{\tau,(F,G)}
=
2\dot\Theta_{\tau,(F,G)}
-
4\dot A_{(F,G)}.
\]
where $\dot\Theta$ is the derivative from Lemma \ref{lemma1} and $\dot A$ is given in Lemma~\ref{lemma4}
\end{lemma}
\begin{proof}
    Since $G(m)=1/2$,
\[
\theta^{\le}(F,G)=\frac{A(F,G)}{G(m)}=2A(F,G),\qquad
\theta^{>}(F,G)=\frac{\Theta(F,G)-A(F,G)}{1-G(m)}=2(\Theta(F,G)-A(F,G)),
\]
hence $O=\theta^>-\theta^{\le}=2\Theta-4A$.
Hadamard differentiability follows because sums and scalar multiples of Hadamard
differentiable maps are Hadamard differentiable, and the derivative is the
corresponding linear combination of derivatives.
\end{proof}
For every fixed $m<\tau<\tau_0$, define the localized effect vector
\[
\boldsymbol\psi_\tau(F,G)
:=
\begin{pmatrix}
\Theta_\tau(F,G)\\
O_\tau(F,G)
\end{pmatrix}.
\]

For fixed $\tau<\tau_0$, let
\[
\widehat{\boldsymbol\psi}_{N,\tau}
:=
\boldsymbol\psi_\tau(\widehat F,\widehat G).
\]
The estimator used in the main paper is denoted by
$\widehat{\boldsymbol\psi}_N$ and is defined in Section~\ref{section2.3}.

\subsection{Proof of Theorem \ref{thm:asym}} \label{proofthm1}
\begin{proof}

Fix $m<\tau<\tau_0$.
By Lemmas  \ref{lemma2} and \ref{lemma5}, the localized map
$\boldsymbol\psi_\tau$ is Hadamard differentiable at $(F,G)$,
tangentially to $\mathbb H$, with continuous linear derivative
$\dot{\boldsymbol\psi}_{\tau,(F,G)}$.

By \eqref{KM_CLT},
\[
\sqrt N\binom{\widehat F-F}{\widehat G-G}\rightsquigarrow \mathbb Z
\quad\text{in }D[0,\tau]^2.
\]

Since the limit paths of $\mathbb Z$ lie in $\mathbb H$ almost surely, the functional delta method yields
\[
\sqrt N
\left\{
\widehat{\boldsymbol\psi}_{N,\tau}
-
\boldsymbol\psi_\tau(F,G)
\right\}
\rightsquigarrow
\dot{\boldsymbol\psi}_{\tau,(F,G)}(\boldsymbol Z).
\]
Because $\mathbb Z$ is mean-zero Gaussian and $\dot{\boldsymbol\psi}_{\tau,(F,G)}$ is continuous and linear,
\[
\dot{\boldsymbol\psi}_{\tau,(F,G)}(\boldsymbol Z)
\sim
\mathcal N_2
\left(
\boldsymbol 0,
\boldsymbol\Sigma_\tau(F,G)
\right).
\]

By Assumption~\ref{ass:tail-localization},
\[
\sqrt N
\left\|
\left\{
\widehat{\boldsymbol\psi}_N
-
\boldsymbol\psi(F,G)
\right\}
-
\left\{
\widehat{\boldsymbol\psi}_{N,\tau}
-
\boldsymbol\psi_\tau(F,G)
\right\}
\right\|
\]
is asymptotically negligible as first $N\to\infty$ and subsequently
$\tau\uparrow\tau_0$. Moreover,
\[
\boldsymbol\Sigma_\tau(F,G)
\longrightarrow
\boldsymbol\Sigma(F,G).
\]
The converging-together theorem therefore yields
\[
\sqrt N
\left\{
\widehat{\boldsymbol\psi}_N
-
\boldsymbol\psi(F,G)
\right\}
\rightsquigarrow
\mathcal N_2
\left(
\boldsymbol 0,
\boldsymbol\Sigma(F,G)
\right).
\]
The limiting covariance matrix is finite and positive semidefinite.
% The version for the pooled sample is shown in the exact same way by replacing $F$ and $G$ with $H$ and $\hat{F}$ and $\hat{G}$ with $\hat{H}_1$ and $\hat{H}_2$. The corresponding covariance matrix is $\Sigma(H):=\mathrm{Cov}\bigl(\dot\psi_{(H,H)}(\mathbb Z_0)\bigr)$.
This proves Theorem~\ref{thm:asym}.
\end{proof}
For the stratified bootstrap, define
\[
\widehat{\boldsymbol\psi}_{N,\tau}^{*,\mathrm{str}}
:=
\boldsymbol\psi_\tau
\left(
\widehat F^*,
\widehat G^*
\right).
\]

Under the null hypothesis, let $\widehat H_N$ denote the
Kaplan--Meier estimator based on the pooled sample, and define
\[
\widehat{\boldsymbol\psi}_{N,\tau}^{*,\mathrm{pool}}
:=
\boldsymbol\psi_\tau
\left(
\widehat H_1^*,
\widehat H_2^*
\right),
\]
where $\widehat H_1^*$ and $\widehat H_2^*$ are the Kaplan--Meier
estimators obtained from the two independent pooled-bootstrap
samples.

\begin{assumption}[Bootstrap regularity]
\label{ass:bootstrap}
Suppose that Assumption~\ref{ass:regularity} holds.

\begin{enumerate}[label=(\roman*)]
\item
For every finite \(m<\tau<\tau_0\), the stratified bootstrap
Kaplan--Meier processes satisfy, conditionally on the data,
\[
\sqrt{n_1}(\widehat F^*-\widehat F)
\rightsquigarrow_{\mathbb P}\mathbb G_F,
\qquad
\sqrt{n_2}(\widehat G^*-\widehat G)
\rightsquigarrow_{\mathbb P}\mathbb G_G,
\]
with conditional independence across groups.

\item \label{ass:bootstrap-tail}
For every \(\varepsilon,\eta>0\), the bootstrap tail-localization condition holds
\[
\lim_{\tau\uparrow\tau_0}
\limsup_{N\to\infty}
\Pr\left[
\Pr^*\left\{
\sqrt N
\left\|
(\widehat\psi_N^{*,\mathrm{str}}-\widehat\psi_N)
-
(\widehat\psi_{N,\tau}^{*,\mathrm{str}}
 -\widehat\psi_{N,\tau})
\right\|>\varepsilon
\right\}>\eta
\right]
=0.
\]

\item
Under \(H_0:F=G=:H\), suppose additionally that the observed censored
pairs are exchangeable across groups. A sufficient condition is equality
of the censoring distributions, \(K=J\). Assume that the corresponding
pooled-bootstrap Kaplan--Meier processes satisfy the conditional
functional central limit theorem.

\item \label{ass:bootstrap-tail2}
The pooled-bootstrap tail-localization condition holds, and
\[
\Sigma_{0,\tau}\longrightarrow\Sigma_0
\qquad\text{as }\tau\uparrow\tau_0.
\]
\end{enumerate}
\end{assumption}

\subsection{Proof of Theorem \ref{thm_boot}}  \label{proofthm2}

Let $(\widehat F^*,\widehat G^*)$ be the KM estimators computed from the stratified bootstrap samples,
and set $\widehat\psi_{N}^*:=\psi(\widehat F^*,\widehat G^*)$.

By \eqref{KM_conditional_CLT},
\[
\sqrt N
\begin{pmatrix}
\widehat F^*-\widehat F\\
\widehat G^*-\widehat G
\end{pmatrix}
\rightsquigarrow_{\mathbb P}\mathbb Z
\quad\text{in }D[0,\tau]^2.
\]

By the bootstrap functional delta method applied to the localized
map $\boldsymbol\psi_\tau$,
\[
\sqrt N
\left\{
\widehat\psi_{N,\tau}^{*,\mathrm{str}}
-\widehat\psi_{N,\tau}
\right\}
\rightsquigarrow_{\mathbb P}
\dot\psi_{\tau,(F,G)}(\mathbb Z)
\sim N_2(0,\Sigma_\tau(F,G)).
\]
By Assumption~\ref{ass:bootstrap-tail}, the difference between the
full and localized centered bootstrap statistics is conditionally
negligible as first $N\to\infty$ and subsequently
$\tau\uparrow\tau_0$. Since
\[
\boldsymbol\Sigma_\tau(F,G)
\longrightarrow
\boldsymbol\Sigma(F,G),
\]
the conditional converging-together theorem yields
\[
\sqrt N
\left\{
\widehat\psi_N^{*,\mathrm{str}}-\widehat\psi_N
\right\}
\rightsquigarrow_{\mathbb P}
N_2(0,\Sigma(F,G)).
\]

 Now assume $H_0:F=G=:H$ and equality of the censoring
distributions. The pooled observations are then exchangeable.
Conditionally on the data,
\[
\sqrt N
\begin{pmatrix}
\widehat H_1^*-\widehat H_N\\
\widehat H_2^*-\widehat H_N
\end{pmatrix}
\rightsquigarrow_{\mathbb P}
\mathbb Z_{0,\tau}
\quad\text{in }D[0,\tau]^2.
\]
Therefore,
\[
\sqrt N
\left[
\psi_\tau(\widehat H_1^*,\widehat H_2^*)
-
\psi_\tau(\widehat H_N,\widehat H_N)
\right]
\rightsquigarrow_{\mathbb P}
\dot\psi_{\tau,(H,H)}(\mathbb Z_{0,\tau})
\sim N_2(0,\Sigma_{0,\tau}).
\]
Assumption~\ref{ass:bootstrap-tail2},  together with
$\boldsymbol\Sigma_{0,\tau}\to\boldsymbol\Sigma_0$, now gives
\[
\sqrt N
\left\{
\widehat\psi_N^{*,\mathrm{pool}}-\psi_0
\right\}
\rightsquigarrow_{\mathbb P}
N_2(0,\Sigma_0).
\]
 This proves Theorem~\ref{thm_boot}.
\subsection{Proof of Remark \ref{remark3}} \label{proofremark3}
\begin{proof}
Since $S_0$ is continuous, both $F$ and $G$ are continuous, so that
$F^\pm=F$. Let $U\sim G$ and define
\[
a:=\frac{\lambda_1}{\lambda_2}=\frac{1}{\eta}.
\]
Since $G$ is continuous, $V:=G(U)\sim\mathrm{Unif}(0,1)$. Moreover,
\[
F(U)
 =1-S_0(U)^{\lambda_1}
 =1-\left\{S_0(U)^{\lambda_2}\right\}^{\lambda_1/\lambda_2}
 =1-\{1-G(U)\}^{a}
 =1-(1-V)^a.
\]
Consequently,
\[
\theta(F,G)
 =\mathbb{E}\{F(U)\}
 =\int_0^1\{1-(1-v)^a\}\,dv
 =\frac{a}{a+1}
 =\frac{\lambda_1}{\lambda_1+\lambda_2}
 =\frac{1}{1+\eta}.
\]

Let $m$ be a median of $G$, so that $G(m)=1/2$. Up to a null set,
$\{U\leq m\}=\{V\leq 1/2\}$. Hence
\[
\theta_{\leq}(F,G)
 =2\int_0^{1/2}\{1-(1-v)^a\}\,dv
 =1-\frac{2}{a+1}\left\{1-2^{-(a+1)}\right\},
\]
whereas
\[
\theta_{>}(F,G)
 =2\int_{1/2}^1\{1-(1-v)^a\}\,dv
 =1-\frac{2^{-a}}{a+1}.
\]
Therefore,
\[
O(F,G)
 =\theta_{>}(F,G)-\theta_{\leq}(F,G)
 =\frac{2}{a+1}\left(1-2^{-a}\right).
\]
Substituting $a=1/\eta$ gives
\[
O(F,G)
 =\frac{2\eta}{1+\eta}\left(1-2^{-1/\eta}\right).
\]
\end{proof}
\section{Additional Remarks}  \label{remarks}
\begin{remark}
  The median split used in Definition \ref{definition2} is not essential for the construction. More generally, for a fixed
\(p\in(0,1)\), let
\(
    q_p := G^{-1}(p)
\)
denote the \(p\)-quantile of the event-time distribution in the second group. We may then define
$ %\[
    \theta_{\le p}(F,G)
    :=
    E\{F^\pm(U)\mid U\le q_p\}, 
    %\qquad
    \theta_{>p}(F,G)
    :=
    E\{F^\pm(U)\mid U>q_p\}, 
    %\qquad 
    (\text{where } U\sim G).
$ %\]
This yields the quantile-split temporal contrast
\[
    O_p(F,G)
    :=
    \theta_{>p}(F,G)-\theta_{\le p}(F,G).
\]
The median-split contrast used throughout this paper corresponds to the special case
\(p=1/2\).

The same asymptotic arguments apply for any fixed \(p\in(0,1)\), provided that \(q_p\) is unique,
lies in the identifiable follow-up region, and \(G\) is sufficiently regular in a neighborhood of
\(q_p\), for example with positive density at \(q_p\). In practice, however, the choice of \(p\) should
be made a priori. Values close to 0 or 1 lead to highly unbalanced conditional regions and may inflate
the variance because one of the two denominators, \(p\) or \(1-p\), becomes small. Under right censoring,
additional care is required: the empirical quantile \(q_p\) should be estimable from the Kaplan--Meier
curve, and there should be enough uncensored event times on both sides of the split. In particular,
higher censoring in the tail makes large values of \(p\) less stable or even non-estimable. The median
therefore provides a natural default, since it yields a balanced early/late comparison and is usually
more stable than more extreme quantile splits. Other fixed choices of \(p\) may nevertheless be useful
when there is prior clinical or biological motivation to focus on an earlier or later part of the
follow-up period.
\end{remark}
\begin{remark}
    It is shown in Theorem 2.2 in~\cite{dey2025inference}, that, in the case of $F$ and $G$ both coming from the same scale family of distributions, that is
    $ %\begin{align*}
        %&
        F(t) = H\left( %\frac{t}{\lambda_1} 
        t /\lambda_1
        \right), %\\
        %&
        G(t) = H\left( %\frac{t}{\lambda_2} 
        t/\lambda_2
        \right),
    $ %\end{align*}
    for some cdf $H$, then $O$ is merely a function of $\lambda_1/\lambda_2
    %\frac{\lambda_1}{\lambda_2}
    $. 
    Their proof can be easily extended to show that this is also true for $\theta$ and therefore in this case the whole vector $(\theta, O)^{\top}$ is solely a function of $%\frac{\lambda_1}{\lambda_2}
    \lambda_1/\lambda_2
    $.
\end{remark}
\section{Possible Pairs}
\subsection{Joint Region of $\theta(F,G)$ and $O(F,G)$ }
\begin{remark}
\label{rem:joint_region1}
    While both $\theta(F, G)$ and $O(F, G)$ are elements of $[0, 1]$, not all value pairs in $[0, 1]^2$ are attainable. As is proven in Theorem 1 in~\cite{Beck2023Combining} the attainable region is the interior of the triangle $\{(\theta, O) | \theta \in [0, 1],\, O \in [0,\min{2 \theta, 2 - 2 \theta} ]\}$
\end{remark}
\subsection{Joint Region of $\theta(F,G)$ and $O(F,G)+
O(G,F)$ }
\begin{remark}
\label{rem:joint_region2}
    For $\theta$ we have $\theta(F, G) + \theta(G, F) = 1$. It is shown in Lemma 1 in~\cite{parkinson2022testing}, that for continuous $F$ and $G$, in the case of $F^{-1}(1/2) = G^{-1}(1/2)$ we have $O(F, G) + O(G, F) = 1$. In general however we can only guarantee $O(F, G) + O(G, F) \leq 1$. More precisely we have the following
\end{remark}

\begin{proposition}[A lower bound for $O(F,G)+O(G,F)$ in terms of $\theta$]
Let $X\sim F$ and $Y\sim G$ be independent with continuous distribution functions and let
\begin{equation*}
m_F := F^{-1}(1/2), \qquad m_G := G^{-1}(1/2)
\end{equation*}
denote the medians of $F$ and $G$, respectively. Then
\begin{equation*}
\frac{4}{3}\min\{\theta,\,1-\theta\} \leq O(F,G) + O(G,F) \leq  \min\{4 \theta, 1, 4 - 4 \theta\}.
\end{equation*}
More precisely, as a lower bound we have
\begin{equation*}
O(F,G)+O(G,F) \ge
\begin{cases}
\dfrac{4}{3}\theta, & m_F>m_G,\\[2mm]
1, & m_F=m_G,\\[2mm]
\dfrac{4}{3}(1-\theta), & m_F<m_G.
\end{cases}
\end{equation*}
\end{proposition}

\begin{proof}
The upper bound immediately follows by adding up the upper bounds for each $O(F, G)$ and $O(G, F)$ given in Remark~\ref{rem:joint_region1}.

In order to prove the lower bound Let $X^{(1)},X^{(2)},Y^{(1)},Y^{(2)}$ be independent versions of $X$ and $Y$ conditional on being smaller/larger than $m_F$/$m_G$ respectively. That is
\begin{align*}
X^{(1)} \sim F^{(1)}(t) = \begin{cases} 2F(t) & t \leq m_F\\
1 & t > m_F
\end{cases}
\quad
X^{(2)} \sim F^{(2)}(t) = \begin{cases} 0 & t \leq m_F\\
2F(t) - 1 & t > m_F
\end{cases}
\end{align*}
and analogously for $Y^{(1)}$ and $Y^{(2)}$.

Define
\begin{align*}
p_{11} := P\bigl(X^{(1)}<Y^{(1)}\bigr), \qquad
p_{12} := P\bigl(X^{(1)}<Y^{(2)}\bigr), \qquad
p_{21} := P\bigl(X^{(2)}<Y^{(1)}\bigr), \qquad
p_{22} := P\bigl(X^{(2)}<Y^{(2)}\bigr).
\end{align*}
$\theta$ and $O(F, G)$ can then be represented as
\begin{equation*}
\theta=\frac{p_{11} + p_{12} + p_{21} + p_{22}}{4},
\qquad
O(F,G)=\frac{p_{12} + p_{22} - p_{11} - p_{21}}{2}.
\end{equation*}
On the other hand, for $O(G, F)$ we obtain
\begin{align*}
O(G, F) &= \frac{1}{2} \left[ P(Y^{(1)} < X^{(2)}) + P(Y^{(2)} < X^{(2)}) - P(Y^{(1)} < X^{(1)}) - P(Y^{(2)} < X^{(1)})  \right]\\
&= \frac{1 - p_{21} + 1 - p_{22} - (1 - p_{11}) - (1 - p_{12})}{2} = \frac{p_{11} + p_{12} - p_{21} - p_{22}}{2}
\end{align*}
Therefore,
\begin{equation*}
O(F,G) + O(G,F) = p_{12} - p_{21}.
\end{equation*}

We now distinguish the three possible orders of the medians.

\medskip
\noindent
\textbf{Case 1: $m_F>m_G$.}
Since $X^{(2)}$ is supported on $(m_F,\infty)$ and $Y^{(1)}$ is supported on $(-\infty,m_G]$,
we have
\begin{equation*}
X^{(2)} > Y^{(1)} \quad \text{a.s.},
\end{equation*}
and thus
\begin{equation*}
p_{21} = P\bigl(X^{(2)} < Y^{(1)}\bigr) = 0.
\end{equation*}
Consequently,
\begin{equation*}
O(F,G) + O(G,F) = p_{12}.
\end{equation*}
Moreover, since $Y^{(2)}$ is stochastically larger than $Y^{(1)}$,
\begin{equation*}
p_{11} \leq p_{12},
\end{equation*}
and since $X^{(1)}$ is stochastically smaller than $X^{(2)}$,
\begin{equation*}
p_{22} \leq p_{12}.
\end{equation*}
Using $\theta=(p_{11} + p_{12} + p_{21} + p_{22})/4$ and $p_{21} = 0$, we get
\begin{equation*}
4\theta = p_{11} + p_{12} + p_{22} \leq 3 p_{12} = 3 \bigl(O(F,G) + O(G,F)\bigr).
\end{equation*}
Therefore,
\begin{equation*}
O(F,G) + O(G,F) \geq \frac{4}{3}\theta.
\end{equation*}

\medskip
\noindent
\textbf{Case 2: $m_F<m_G$.}
Since $X^{(1)}$ is supported on $(-\infty,m_F]$ and $Y^{(2)}$ is supported on $(m_G,\infty)$,
we have
\begin{equation*}
X^{(1)} < Y^{(2)} \quad \text{a.s.},
\end{equation*}
and hence
\begin{equation*}
p_{12} = P\bigl(X^{(1)}<Y^{(2)}\bigr) = 1.
\end{equation*}
Thus
\begin{equation*}
O(F,G) + O(G,F) = 1 - p_{21}.
\end{equation*}
Further, because $X^{(1)}$ is stochastically smaller than $X^{(2)}$,
\begin{equation*}
p_{11} \geq p_{21},
\end{equation*}
and because $Y^{(2)}$ is stochastically larger than $Y^{(1)}$,
\begin{equation*}
p_{22} \geq p_{21}.
\end{equation*}
Therefore,
\begin{equation*}
4\theta = p_{11} + p_{12} + p_{21} + p_{22} = p_{11} + 1 + p_{21} + p_{22} \geq 1 + 3p_{21}.
\end{equation*}
Hence
\begin{equation*}
p_{21} \leq \frac{4\theta - 1}{3},
\end{equation*}
and so
\begin{equation*}
O(F,G) + O(G,F) = 1 - p_{21} \geq 1 - \frac{4\theta - 1}{3}
= \frac{4}{3}(1 - \theta).
\end{equation*}

\medskip
\noindent
\textbf{Case 3: $m_F=m_G$.}
In this case, it is shown in \cite{ParkinsonKutilKupplerJunkerTrutschnigBathke+2018} that
\begin{equation*}
O(F,G) + O(G,F) = 1.
\end{equation*}

Combining the three cases proves the piecewise lower bound
\begin{equation*}
O(F,G) + O(G,F) \geq
\begin{cases}
\dfrac{4}{3}\theta, & m_F>m_G,\\[2mm]
1, & m_F=m_G,\\[2mm]
\dfrac{4}{3}(1-\theta), & m_F<m_G,
\end{cases}
\end{equation*}
and therefore, in all cases,
\begin{equation*}
O(F,G) + O(G,F) \geq \frac{4}{3}\min\{\theta,\,1-\theta\}.
\end{equation*}
This completes the proof.
\end{proof}

\begin{figure}
    \centering
    \includegraphics[width=0.8\linewidth]{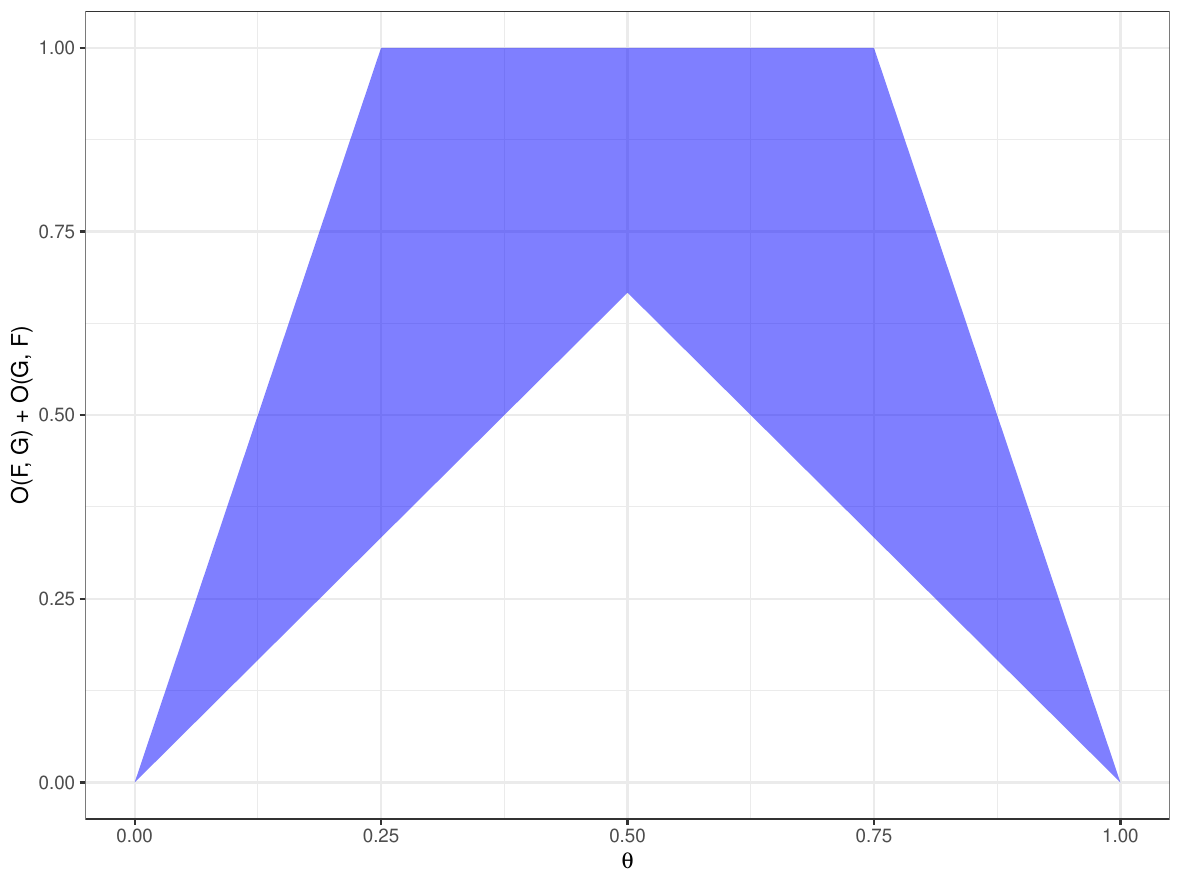}
    \caption{Illustration of the joint region of $\theta(F, G)$ and $O(F, G) + O(G, F)$. For each pair $(F, G)$ of cumulative distribution functions the pair $(\theta(F, G),\, O(F, G) + O(G, F))$ must lie within the blue shaded region.}
    \label{fig:placeholder}
\end{figure}

\section{Simulation setting} \label{sec_simu}

In order to illustrate our procedure and compare it with competitors we conducted a simulation study. Data are simulated using the R package \texttt{rsurv} (\cite{rsurv}), specifically using the function \texttt{rypreg}, which generates data from a Yang and Prentice model (\cite{yang2005semiparametric}). A survival function from this model takes the form

\begin{equation*}
    S(t| \beta, \varphi) = \left( 1 + \frac{1 - S_0(t)}{S_0(t)} e^{\beta - \varphi}  \right)^{-e^{-\varphi}}
\end{equation*}
where $S_0(t)$ is a baseline survival function and $\beta, \varphi \in \mathbb{R}$ are the model parameters. The survival curves cross at some point, when $\beta \cdot \varphi < 0$ and in the case of $\beta = \varphi$ the model simplifies to the proportional hazards model. This form of survival function will be used in some of the following scenarios.

We considered five basic scenarios
\begin{enumerate}
    \item The \emph{null} scenario, where both samples are drawn from the same distribution. We used three different distributions: An exponential distribution with rate $1$ ($\mathcal{E}(1)$), a Weibull distribution with shape parameter $2$ and scale parameter $1$ ($Weibull(2, 1)$) and a gamma distribution with shape parameter $0.5$ and scale parameter $1$ ($Gamma(0.5, 1)$).
    \item The \emph{proportional hazards} scenario, where the first sample is drawn from $\mathcal{E}(1)$ and the other from $\mathcal{E}(1.5), \mathcal{E}(2)$ and $\mathcal{E}(3)$ respectively.
    \item A \emph{crossing-in-the-middle} scenario from the YP model, where the first sample is drawn according to $S_0(t)$ which is the survival function of a $Weibull(2, 1)$ distribution, and the other according to $S(t | 1, -1)$.
    \item A \emph{crossing-at-the-end} scenario from the YP model, where the first sample is drawn according to $S_0(t)$ which is the survival function of a $Gamma(0.5, 1)$ distribution, and the other according to $S(t | -1, 1)$.
    \item \emph{Diffuse transition alternative (DTA)}. Here, the control group follows an exponential survival model with constant hazard $h_G(t)=1$.
    The treatment hazard is piecewise constant, with \[ h_F(t)=c_j \quad \text{for } t_{j-1}<t\le t_j, \] where \[ (t_j)_{j=0}^7=(0,0.5,1.0,1.6,2.2,3.0,4.0,\infty) \] and \[ (c_j)_{j=1}^7=(0.50,0.65,0.85,1.05,1.25,1.55,1.85). \]
Thus, treatment is beneficial early, approximately neutral during an
intermediate period, and harmful later. In contrast to a sharp crossing-hazard
alternative, the treatment effect changes gradually over time. This setting is
intended to represent a diffuse temporal transition from early benefit to late
harm, for which a single crossing point is not the most natural description of.
\end{enumerate}
For all scenarios, the sample sizes used in the simulation are divided into the balanced and unbalanced cases. For the balanced case, the sample sizes satisfy $n_1 = n_2$ and are given by 30, 50, 100 and 200. For the unbalanced case, we used the following sample sizes: $(n_1, n_2) = (30, 60), (n_1, n_2) = (50, 100), (n_1, n_2) = (100, 200), (n_1, n_2) = (200, 400)$.

We ran each scenario $1000$ times, using $1000$ bootstrap samples and $100,000$ Monte-Carlo samples for estimating the distribution of the max-t-type statistic $T_N^{max}$.
As a censoring distribution we used exponential $\mathcal{E}(\lambda_c)$. For all scenarios except the diffuse transition alternative, $\lambda_c$ was chosen as
\begin{equation*}
    \lambda_c = \frac{1 - r}{r \mu}.
\end{equation*}
Here $\mu$ is the mean survival time in each scenario, and $r$ is the target event rate. For the DTA scenario we instead determined $\lambda_c$ by solving the equation
\begin{equation*}
    \mathbb{P}(T < K) = r
\end{equation*}
where the event time $T$ is simulated as described above and $K {\sim} \mathcal{E}(\lambda_c)$.
We simulated each scenario with $r = 0.7$ and $r = 0.9$.

As competitors we used
\begin{enumerate}
    \item The classic log-rank test.
    \item The Gehan-Wilcoxon test, i.e. a weighted log-rank test putting more weight on earlier observations.
    \item A weighted log-rank test with data-dependent adaptive weights (\cite{yangprentice09}).
    \item A two-stage procedure, which splits the original $\alpha$ available, using a portion of it for a log-rank test and in the case this does not reject, using the rest to test for a crossing (\cite{qiusheng08}).
\end{enumerate}
\section{Additional Example}
\subsection{Gastrointestinal Cancer}
\label{app:additional_example}
As an additional real-data example we use data from the Gastrointestinal Cancer Study Group~\cite{schein1982comparison}. The data stem from a randomized two-arm trial, assigning patients with locally advanced gastric carcinoma to either chemotherapy or chemotherapy plus radiation therapy. Each group contained $45$ subjects. Almost all patients were followed up until the event, with only $2$ patients in the chemotherapy group and $6$ patients in the chemotherapy plus radiation group being lost to follow-up. In Figure~\ref{fig:gastric_example} the survival curves of the two groups and the $95\%$-Wald-type confidence ellipse based on the stratified bootstrap covariance estimate can be seen. For this analysis we will set $F = F_{\text{Chemo}}$ and $G = G_{\text{Chemo + Radiation}}$. We obtain the estimate $\hat{\psi} = (0.37, 0.54)$, globally favoring chemotherapy without additional radiation therapy. Analyzing the temporal contrast in detail we see that
\begin{equation*}
    \widehat{\theta}_{\leq} = 0.1 \qquad \widehat{\theta}_{>} = 0.64
\end{equation*}
and the difference from what we would expect under the null is
\begin{equation*}
    \widehat{\delta}_E = \widehat{\theta}_{\leq} - \frac{1}{4} = -0.15 \qquad \widehat{\delta}_{L} = \widehat{\theta}_{>} - \frac{3}{4} = -0.11
\end{equation*}
which suggests that chemotherapy without radiation has an advantage early and late. Note that early and late here are with respect to the median of the chemotherapy plus radiation group (approximately $1.5$ years) so this result is not in conflict with the chemotherapy plus radiation group appearing to be better later still, from approximately $3$ years on. The max-t-type test gives a p-value of $0.058$ so the results should be interpreted descriptively.

%What we would like to particularly focus on here is the fact that the survival curves cross and the addition of radiation therapy to chemotherapy seems to be associated with an early decrease in survival time, but a higher probability of survival long-term. This behavior is not unusual in cancer research, since more radical treatment increases mortality early on, but also decreases chance of cancer recurrence. A particular challenge in survival analysis is detecting this type of deviation from the null hypothesis of equality of the survival curves. Some standard statistical procedures, like testing for the hazard ratio in a Cox proportional hazard model are not at all sensitive to this alternative, while some, like the logrank test are, but have reduced power for this alternative compared to a proportional hazards setting.

\begin{figure}
    \centering
    \begin{subfigure}[t]{0.58\linewidth}
        \centering
        \includegraphics[width=\linewidth]{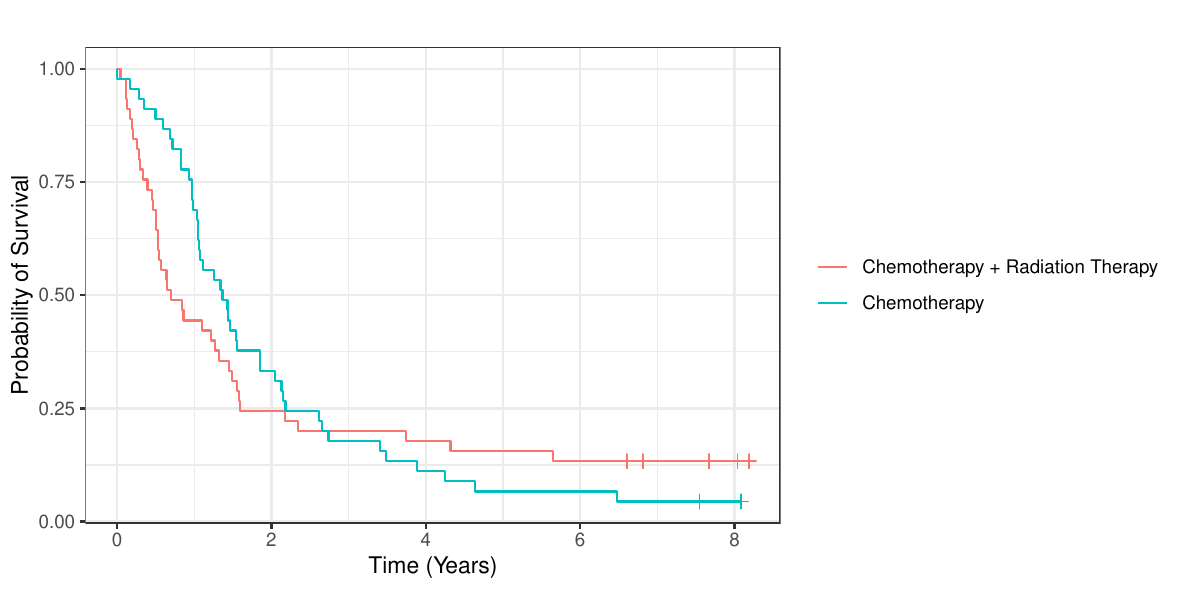}
        \caption{Kaplan--Meier estimates of overall survival.}
        \label{fig:km_gastric}
    \end{subfigure}
    \hfill
    \begin{subfigure}[t]{0.38\linewidth}
        \centering
        \includegraphics[width=\linewidth]{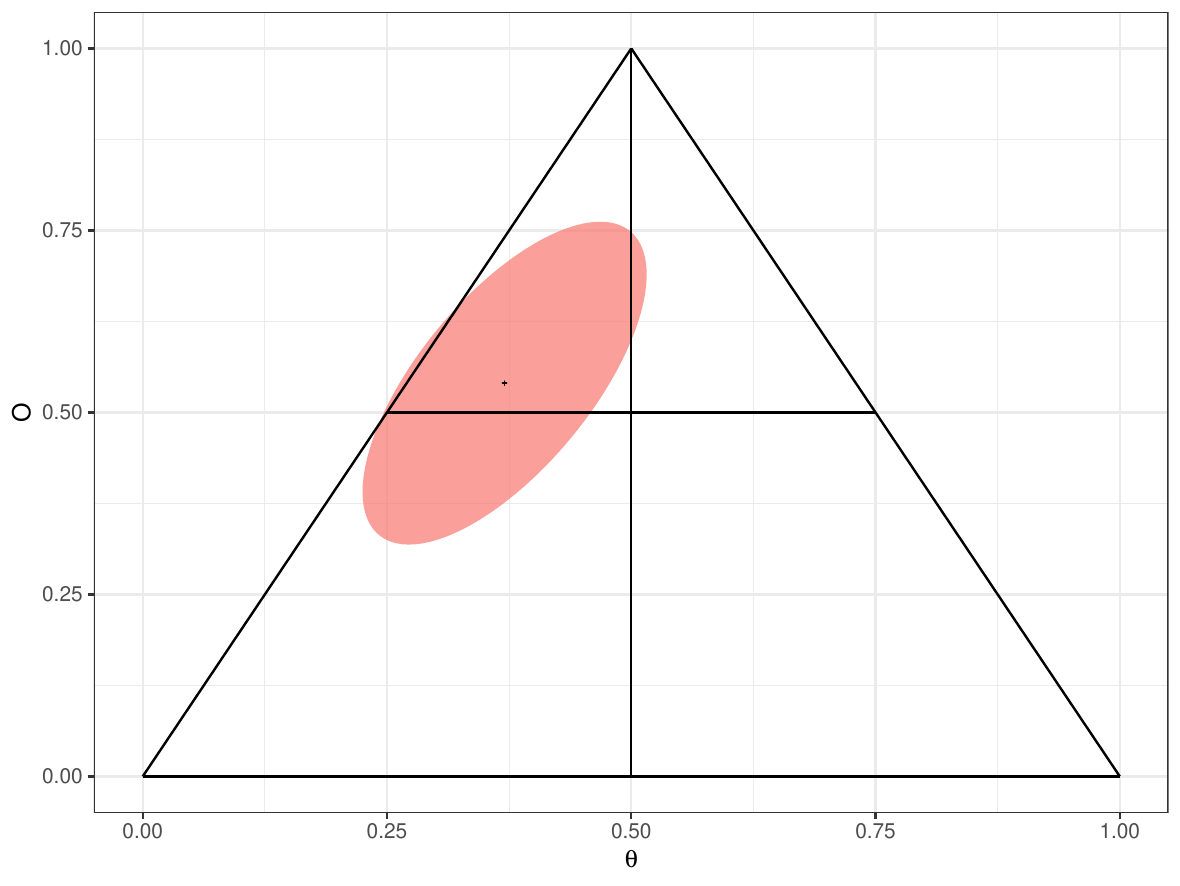}
        \caption{Estimated joint effect with bootstrap uncertainty region.}
        \label{fig:ellipse_gastric}
    \end{subfigure}

    \caption{Overall survival analysis in the gastric cancer example.
    \textbf{(a)} Kaplan--Meier estimates for patients receiving either chemotherapy or chemotherapy plus radiation therapy. The curves suggest a non-proportional pattern, with an early advantage for chemotherapy only that attenuates and turns into a late advantage of additional radiation therapy.
    \textbf{(b)} Estimated joint effect
    \((\widehat\theta,\widehat O)=(0.37,0.54)\). The ellipse shows the
    bootstrap-based joint uncertainty region, and the null value
    \((1/2,1/2)^\top\) is shown for reference.}
    \label{fig:gastric_example}
\end{figure}

\section{More Details about the Simulations}
\begin{figure}
    \centering
    \includegraphics[width=0.9\linewidth]{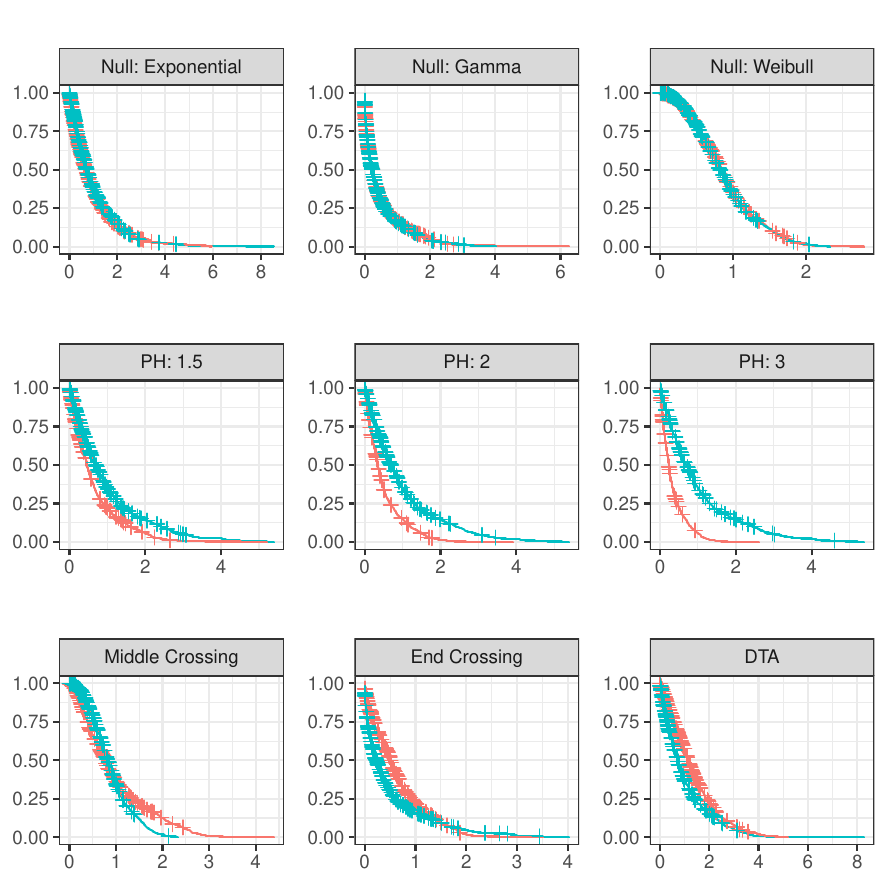}
    \caption{Illustration of the simulation scenarios. Each plot shows the Kaplan-Meier curve based on a sample of $1000$ observations for each group. In the top row the three null scenarios are shown. The middle row shows the proportional hazards scenario for the different hazard rates. The bottom row illustrates the crossing-in-the-middle, crossing-at-the-end and diffuse transition alternative scenarios.}
    \label{fig:sim_illustration}
\end{figure}

\begin{table}[ht]
\caption{\label{tab:null_balanced}Simulation results for the balanced null scenario. The column names specify the distribution family of $F$ and $G$ as well as the event rate. Shown is the rejection rate.}
\centering
\begin{tabular}{llrrrrrr}
&& \multicolumn{2}{c}{Exponential} & \multicolumn{2}{c}{Gamma} & \multicolumn{2}{c}{Weibull} \\ 
  \cmidrule(l){3-8}
  n & Test & 0.7 & 0.9 & 0.7 & 0.9 & 0.7 & 0.9\\
  \hline
  \multirow{8}{4em}{30} & Max - Stratified & 0.085 & 0.059 & 0.106 & 0.064 & 0.057 & 0.049 \\ 
  & Max - Pooled & 0.036 & 0.030 & 0.032 & 0.039 & 0.051 & 0.048 \\ 
  & Wald - Stratified & 0.092 & 0.066 & 0.105 & 0.060 & 0.091 & 0.065 \\ 
  & Wald - Pooled & 0.049 & 0.049 & 0.062 & 0.051 & 0.049 & 0.048 \\ 
  & Logrank & 0.053 & 0.058 & 0.055 & 0.059 & 0.048 & 0.065 \\ 
  & Wilcoxon & 0.050 & 0.051 & 0.044 & 0.053 & 0.047 & 0.049 \\ 
  & Yang and Prentice & 0.065 & 0.071 & 0.067 & 0.071 & 0.064 & 0.075 \\ 
  & Two-Stage & 0.055 & 0.044 & 0.044 & 0.047 & 0.041 & 0.049 \\ 
  \hline
  \multirow{8}{4em}{50}& Max - Stratified & 0.094 & 0.062 & 0.115 & 0.072 & 0.067 & 0.061 \\ 
  & Max - Pooled & 0.056 & 0.054 & 0.054 & 0.056 & 0.058 & 0.050 \\ 
  & Wald - Stratified & 0.090 & 0.073 & 0.111 & 0.082 & 0.077 & 0.076 \\ 
  & Wald - Pooled & 0.060 & 0.057 & 0.058 & 0.053 & 0.053 & 0.050 \\ 
  & Logrank & 0.057 & 0.068 & 0.051 & 0.065 & 0.066 & 0.062 \\ 
  & Wilcoxon & 0.054 & 0.065 & 0.055 & 0.061 & 0.065 & 0.058 \\ 
  & Yang and Prentice & 0.070 & 0.084 & 0.064 & 0.081 & 0.076 & 0.078 \\ 
  & Two-Stage & 0.046 & 0.048 & 0.046 & 0.051 & 0.047 & 0.055 \\ 
  \hline
  \multirow{8}{4em}{100} & Max - Stratified & 0.078 & 0.063 & 0.113 & 0.075 & 0.061 & 0.059 \\ 
  & Max - Pooled & 0.053 & 0.062 & 0.071 & 0.056 & 0.064 & 0.062 \\ 
  & Wald - Stratified & 0.079 & 0.070 & 0.113 & 0.079 & 0.079 & 0.068 \\ 
  & Wald - Pooled & 0.059 & 0.067 & 0.083 & 0.067 & 0.057 & 0.061 \\ 
  & Logrank & 0.061 & 0.062 & 0.062 & 0.057 & 0.054 & 0.066 \\ 
  & Wilcoxon & 0.071 & 0.059 & 0.069 & 0.057 & 0.063 & 0.063 \\ 
  & Yang and Prentice & 0.074 & 0.065 & 0.073 & 0.062 & 0.067 & 0.071 \\ 
  & Two-Stage & 0.055 & 0.062 & 0.054 & 0.058 & 0.054 & 0.059 \\ 
  \hline
  \multirow{8}{4em}{200} & Max - Stratified & 0.064 & 0.054 & 0.091 & 0.054 & 0.050 & 0.044 \\ 
  & Max - Pooled & 0.054 & 0.053 & 0.059 & 0.053 & 0.057 & 0.044 \\ 
  & Wald - Stratified & 0.062 & 0.053 & 0.093 & 0.055 & 0.064 & 0.052 \\ 
  & Wald - Pooled & 0.051 & 0.053 & 0.061 & 0.054 & 0.056 & 0.047 \\ 
  & Logrank & 0.042 & 0.041 & 0.045 & 0.043 & 0.048 & 0.045 \\ 
  & Wilcoxon & 0.038 & 0.038 & 0.038 & 0.037 & 0.040 & 0.036 \\ 
  & Yang and Prentice & 0.044 & 0.048 & 0.054 & 0.050 & 0.055 & 0.049 \\ 
  & Two-Stage & 0.045 & 0.050 & 0.046 & 0.042 & 0.040 & 0.047 \\
   \hline 
\end{tabular}
\end{table}

\begin{table}[ht]
\caption{\label{tab:null_unbalanced}Simulation results for the unbalanced null scenario. The column names specify the distribution family of $F$ and $G$ as well as the event rate. Shown is the rejection rate.}
\centering
\begin{tabular}{llrrrrrr}
&& \multicolumn{2}{c}{Exponential} & \multicolumn{2}{c}{Gamma} & \multicolumn{2}{c}{Weibull} \\ 
  \cmidrule(l){3-8}
  n & Test & 0.7 & 0.9 & 0.7 & 0.9 & 0.7 & 0.9\\
  \hline
  \multirow{8}{4em}{30} & Max - Stratified & 0.123 & 0.087 & 0.179 & 0.102 & 0.065 & 0.078 \\ 
  & Max - Pooled & 0.064 & 0.049 & 0.070 & 0.062 & 0.054 & 0.059 \\ 
  & Wald - Stratified & 0.119 & 0.076 & 0.163 & 0.086 & 0.081 & 0.076 \\ 
  & Wald - Pooled & 0.081 & 0.058 & 0.082 & 0.066 & 0.063 & 0.060 \\ 
  & Logrank & 0.059 & 0.057 & 0.055 & 0.059 & 0.056 & 0.061 \\ 
  & Wilcoxon & 0.056 & 0.057 & 0.061 & 0.061 & 0.058 & 0.063 \\ 
  & Yang and Prentice & 0.068 & 0.073 & 0.071 & 0.077 & 0.082 & 0.078 \\ 
  & Two-Stage & 0.056 & 0.048 & 0.057 & 0.054 & 0.056 & 0.051 \\ 
  \hline
  \multirow{8}{4em}{50} & Max - Stratified & 0.113 & 0.075 & 0.184 & 0.080 & 0.057 & 0.063 \\ 
  & Max - Pooled & 0.057 & 0.068 & 0.074 & 0.062 & 0.058 & 0.059 \\ 
  & Wald - Stratified & 0.109 & 0.078 & 0.168 & 0.080 & 0.079 & 0.071 \\ 
  & Wald - Pooled & 0.071 & 0.073 & 0.090 & 0.070 & 0.055 & 0.060 \\ 
  & Logrank & 0.062 & 0.062 & 0.065 & 0.061 & 0.064 & 0.062 \\ 
  & Wilcoxon & 0.060 & 0.060 & 0.055 & 0.057 & 0.058 & 0.062 \\ 
  & Yang and Prentice & 0.072 & 0.076 & 0.078 & 0.072 & 0.079 & 0.070 \\ 
  & Two-Stage & 0.063 & 0.056 & 0.063 & 0.051 & 0.057 & 0.061 \\ 
  \hline
  \multirow{8}{4em}{100} & Max - Stratified & 0.104 & 0.072 & 0.159 & 0.078 & 0.058 & 0.059 \\ 
  & Max - Pooled & 0.071 & 0.068 & 0.077 & 0.063 & 0.063 & 0.053 \\ 
  & Wald - Stratified & 0.092 & 0.066 & 0.142 & 0.075 & 0.066 & 0.065 \\ 
  & Wald - Pooled & 0.071 & 0.061 & 0.083 & 0.063 & 0.059 & 0.054 \\ 
  & Logrank & 0.062 & 0.057 & 0.055 & 0.061 & 0.055 & 0.058 \\ 
  & Wilcoxon & 0.063 & 0.064 & 0.063 & 0.060 & 0.057 & 0.063 \\ 
  & Yang and Prentice & 0.071 & 0.063 & 0.062 & 0.066 & 0.061 & 0.068 \\ 
  & Two-Stage & 0.054 & 0.051 & 0.053 & 0.054 & 0.046 & 0.060 \\ 
  \hline
  \multirow{8}{4em}{200} & Max - Stratified & 0.066 & 0.052 & 0.127 & 0.052 & 0.047 & 0.052 \\ 
  & Max - Pooled & 0.053 & 0.052 & 0.075 & 0.049 & 0.051 & 0.051 \\ 
  & Wald - Stratified & 0.064 & 0.048 & 0.137 & 0.049 & 0.053 & 0.053 \\ 
  & Wald - Pooled & 0.056 & 0.057 & 0.085 & 0.059 & 0.049 & 0.053 \\ 
  & Logrank & 0.046 & 0.047 & 0.046 & 0.047 & 0.047 & 0.049 \\ 
  & Wilcoxon & 0.041 & 0.036 & 0.037 & 0.041 & 0.043 & 0.041 \\ 
  & Yang and Prentice & 0.051 & 0.048 & 0.052 & 0.052 & 0.051 & 0.055 \\ 
  & Two-Stage & 0.046 & 0.046 & 0.043 & 0.049 & 0.056 & 0.051 \\ 
  \hline
\end{tabular}
\end{table}

\begin{table}[ht]
\caption{\label{tab:ph_balanced}Simulation results for the balanced proportional hazards scenario. The column names specify the hazard ratio as well as the event rate. Shown is the rejection rate.}
\centering
\begin{tabular}{llrrrrrr}
&& \multicolumn{2}{c}{1.5} & \multicolumn{2}{c}{2.0} & \multicolumn{2}{c}{3.0} \\ 
  \cmidrule(l){3-8}
  n & Test & 0.7 & 0.9 & 0.7 & 0.9 & 0.7 & 0.9\\
  \hline
  \multirow{8}{4em}{30} & Max - Stratified & 0.226 & 0.269 & 0.552 & 0.624 & 0.903 & 0.941 \\ 
  & Max - Pooled & 0.148 & 0.208 & 0.429 & 0.534 & 0.862 & 0.909 \\ 
  & Wald - Stratified & 0.237 & 0.258 & 0.522 & 0.600 & 0.900 & 0.927 \\ 
  & Wald - Pooled & 0.127 & 0.205 & 0.418 & 0.552 & 0.862 & 0.917 \\ 
  & Logrank & 0.262 & 0.319 & 0.660 & 0.722 & 0.973 & 0.985 \\ 
  & Wilcoxon & 0.231 & 0.258 & 0.567 & 0.609 & 0.927 & 0.941 \\ 
  & Yang and Prentice & 0.293 & 0.357 & 0.685 & 0.752 & 0.979 & 0.990 \\ 
  & Two-Stage & 0.195 & 0.226 & 0.567 & 0.622 & 0.951 & 0.969 \\ 
  \hline
  \multirow{8}{4em}{50} & Max - Stratified & 0.348 & 0.427 & 0.756 & 0.821 & 0.989 & 0.992 \\ 
  & Max - Pooled & 0.281 & 0.354 & 0.706 & 0.795 & 0.986 & 0.992 \\ 
  & Wald - Stratified & 0.354 & 0.395 & 0.748 & 0.812 & 0.988 & 0.992 \\ 
  & Wald - Pooled & 0.266 & 0.366 & 0.687 & 0.809 & 0.988 & 0.993 \\ 
  & Logrank & 0.454 & 0.503 & 0.871 & 0.914 & 0.997 & 0.997 \\ 
  & Wilcoxon & 0.383 & 0.432 & 0.795 & 0.824 & 0.990 & 0.989 \\ 
  & Yang and Prentice & 0.488 & 0.535 & 0.887 & 0.922 & 0.997 & 0.997 \\ 
  & Two-Stage & 0.364 & 0.419 & 0.808 & 0.863 & 0.996 & 0.995 \\ 
  \hline
  \multirow{8}{4em}{100} & Max - Stratified & 0.539 & 0.641 & 0.955 & 0.972 & 1.000 & 1.000 \\ 
  & Max - Pooled & 0.493 & 0.613 & 0.948 & 0.971 & 1.000 & 1.000 \\ 
  & Wald - Stratified & 0.526 & 0.616 & 0.948 & 0.967 & 1.000 & 1.000 \\ 
  & Wald - Pooled & 0.492 & 0.643 & 0.953 & 0.975 & 1.000 & 1.000 \\ 
  & Logrank & 0.712 & 0.791 & 0.986 & 0.994 & 1.000 & 1.000 \\ 
  & Wilcoxon & 0.621 & 0.688 & 0.970 & 0.977 & 1.000 & 1.000 \\ 
  & Yang and Prentice & 0.727 & 0.808 & 0.986 & 0.995 & 1.000 & 1.000 \\ 
  & Two-Stage & 0.601 & 0.713 & 0.982 & 0.987 & 1.000 & 1.000 \\ 
  \hline
  \multirow{8}{4em}{200} & Max - Stratified & 0.855 & 0.918 & 0.999 & 0.999 & 1.000 & 1.000 \\ 
  & Max - Pooled & 0.845 & 0.915 & 0.999 & 0.999 & 1.000 & 1.000 \\ 
  & Wald - Stratified & 0.849 & 0.915 & 0.998 & 0.999 & 1.000 & 1.000 \\ 
  & Wald - Pooled & 0.835 & 0.923 & 0.999 & 0.999 & 1.000 & 1.000 \\ 
  & Logrank & 0.956 & 0.981 & 1.000 & 1.000 & 1.000 & 1.000 \\ 
  & Wilcoxon & 0.907 & 0.938 & 1.000 & 1.000 & 1.000 & 1.000 \\ 
  & Yang and Prentice & 0.960 & 0.981 & 1.000 & 1.000 & 1.000 & 1.000 \\ 
  & Two-Stage & 0.918 & 0.959 & 1.000 & 1.000 & 1.000 & 1.000 \\
   \hline
\end{tabular}
\end{table}

\begin{table}[ht]
\caption{\label{tab:ph_unbalanced}Simulation results for the unbalanced proportional hazards scenario. The column names specify the hazard ratio as well as the event rate. Shown is the rejection rate.}
\centering
\begin{tabular}{llrrrrrr}
&& \multicolumn{2}{c}{1.5} & \multicolumn{2}{c}{2.0} & \multicolumn{2}{c}{3.0} \\ 
  \cmidrule(l){3-8}
  n & Test & 0.7 & 0.9 & 0.7 & 0.9 & 0.7 & 0.9\\
  \hline
  \multirow{8}{4em}{30} & Max - Stratified & 0.355 & 0.363 & 0.705 & 0.743 & 0.975 & 0.985 \\ 
  & Max - Pooled & 0.273 & 0.314 & 0.667 & 0.699 & 0.970 & 0.980 \\ 
  & Wald - Stratified & 0.340 & 0.342 & 0.688 & 0.728 & 0.971 & 0.979 \\ 
  & Wald - Pooled & 0.271 & 0.332 & 0.668 & 0.721 & 0.978 & 0.987 \\ 
  & Logrank & 0.349 & 0.407 & 0.787 & 0.826 & 0.996 & 0.996 \\ 
  & Wilcoxon & 0.294 & 0.326 & 0.700 & 0.727 & 0.974 & 0.977 \\ 
  & Yang and Prentice & 0.378 & 0.433 & 0.803 & 0.839 & 0.995 & 0.997 \\ 
  & Two-Stage & 0.278 & 0.325 & 0.691 & 0.746 & 0.985 & 0.991 \\ 
  \hline
  \multirow{8}{4em}{50} & Max - Stratified & 0.459 & 0.492 & 0.869 & 0.906 & 0.994 & 0.998 \\ 
  & Max - Pooled & 0.426 & 0.470 & 0.867 & 0.900 & 0.994 & 0.998 \\ 
  & Wald - Stratified & 0.451 & 0.476 & 0.857 & 0.896 & 0.995 & 0.998 \\ 
  & Wald - Pooled & 0.428 & 0.504 & 0.861 & 0.909 & 0.996 & 0.998 \\ 
  & Logrank & 0.573 & 0.612 & 0.939 & 0.957 & 0.999 & 0.999 \\ 
  & Wilcoxon & 0.458 & 0.507 & 0.886 & 0.907 & 0.998 & 0.998 \\ 
  & Yang and Prentice & 0.591 & 0.632 & 0.942 & 0.960 & 0.999 & 0.999 \\ 
  & Two-Stage & 0.460 & 0.531 & 0.904 & 0.927 & 0.998 & 0.999 \\ 
  \hline
  \multirow{8}{4em}{100} & Max - Stratified & 0.699 & 0.777 & 0.983 & 0.991 & 1.000 & 1.000 \\ 
  & Max - Pooled & 0.706 & 0.779 & 0.987 & 0.992 & 1.000 & 1.000 \\ 
  & Wald - Stratified & 0.692 & 0.766 & 0.975 & 0.991 & 1.000 & 1.000 \\ 
  & Wald - Pooled & 0.704 & 0.796 & 0.985 & 0.992 & 1.000 & 1.000 \\ 
  & Logrank & 0.831 & 0.882 & 0.997 & 1.000 & 1.000 & 1.000 \\ 
  & Wilcoxon & 0.740 & 0.787 & 0.991 & 0.994 & 1.000 & 1.000 \\ 
  & Yang and Prentice & 0.837 & 0.890 & 0.997 & 1.000 & 1.000 & 1.000 \\ 
  & Two-Stage & 0.756 & 0.839 & 0.996 & 0.998 & 1.000 & 1.000 \\ 
  \hline
  \multirow{8}{4em}{200} & Max - Stratified & 0.939 & 0.969 & 1.000 & 1.000 & 1.000 & 1.000 \\ 
  & Max - Pooled & 0.944 & 0.972 & 1.000 & 1.000 & 1.000 & 1.000 \\ 
  & Wald - Stratified & 0.928 & 0.967 & 1.000 & 1.000 & 1.000 & 1.000 \\ 
  & Wald - Pooled & 0.939 & 0.976 & 1.000 & 1.000 & 1.000 & 1.000 \\ 
  & Logrank & 0.985 & 0.993 & 1.000 & 1.000 & 1.000 & 1.000 \\ 
  & Wilcoxon & 0.965 & 0.974 & 1.000 & 1.000 & 1.000 & 1.000 \\ 
  & Yang and Prentice & 0.985 & 0.992 & 1.000 & 1.000 & 1.000 & 1.000 \\ 
  & Two-Stage & 0.977 & 0.987 & 1.000 & 1.000 & 1.000 & 1.000 \\ 
   \hline
\end{tabular}
\end{table}

\begin{table}[ht]
\caption{\label{tab:crossing_balanced}Simulation results for the balanced crossing scenario. The column names specify the location of the crossing as well as the event rate. Shown is the rejection rate.}
\centering
\begin{tabular}{llrrrrrr}
&& \multicolumn{2}{c}{Mid-Crossing} & \multicolumn{2}{c}{End-Crossing} & \multicolumn{2}{c}{DTA} \\ 
  \cmidrule(l){3-8}
  n & Test & 0.7 & 0.9 & 0.7 & 0.9 & 0.7 & 0.9\\
  \hline
  \multirow{8}{4em}{30} & Max - Stratified & 0.083 & 0.099 & 0.367 & 0.354 & 0.314 & 0.326 \\ 
  & Max - Pooled & 0.127 & 0.150 & 0.249 & 0.310 & 0.198 & 0.260 \\ 
  & Wald - Stratified & 0.183 & 0.178 & 0.383 & 0.382 & 0.331 & 0.338 \\ 
  & Wald - Pooled & 0.122 & 0.147 & 0.368 & 0.380 & 0.258 & 0.304 \\ 
  & Logrank & 0.062 & 0.091 & 0.279 & 0.203 & 0.237 & 0.241 \\ 
  & Wilcoxon & 0.063 & 0.067 & 0.387 & 0.379 & 0.303 & 0.330 \\ 
  & Yang and Prentice & 0.227 & 0.337 & 0.352 & 0.323 & 0.287 & 0.298 \\ 
  & Two-Stage & 0.267 & 0.299 & 0.297 & 0.322 & 0.220 & 0.255 \\ 
  \hline
  \multirow{8}{4em}{50} & Max - Stratified & 0.182 & 0.192 & 0.478 & 0.494 & 0.472 & 0.484 \\ 
  & Max - Pooled & 0.258 & 0.280 & 0.398 & 0.472 & 0.387 & 0.446 \\ 
  & Wald - Stratified & 0.332 & 0.302 & 0.525 & 0.515 & 0.483 & 0.497 \\ 
  & Wald - Pooled & 0.264 & 0.295 & 0.535 & 0.549 & 0.467 & 0.495 \\ 
  & Logrank & 0.078 & 0.117 & 0.402 & 0.278 & 0.386 & 0.363 \\ 
  & Wilcoxon & 0.094 & 0.086 & 0.572 & 0.548 & 0.502 & 0.523 \\ 
  & Yang and Prentice & 0.362 & 0.485 & 0.505 & 0.466 & 0.459 & 0.466 \\ 
  & Two-Stage & 0.432 & 0.522 & 0.469 & 0.523 & 0.395 & 0.399 \\ 
  \hline
  \multirow{8}{4em}{100} & Max - Stratified & 0.500 & 0.587 & 0.753 & 0.816 & 0.699 & 0.758 \\ 
  & Max - Pooled & 0.571 & 0.686 & 0.734 & 0.809 & 0.674 & 0.751 \\ 
  & Wald - Stratified & 0.661 & 0.686 & 0.839 & 0.831 & 0.727 & 0.747 \\ 
  & Wald - Pooled & 0.623 & 0.687 & 0.848 & 0.865 & 0.750 & 0.774 \\ 
  & Logrank & 0.082 & 0.212 & 0.670 & 0.458 & 0.652 & 0.571 \\ 
  & Wilcoxon & 0.142 & 0.088 & 0.866 & 0.858 & 0.800 & 0.809 \\ 
  & Yang and Prentice & 0.589 & 0.803 & 0.808 & 0.774 & 0.744 & 0.733 \\ 
  & Two-Stage & 0.815 & 0.913 & 0.767 & 0.821 & 0.641 & 0.665 \\ 
  \hline
  \multirow{8}{4em}{200} & Max - Stratified & 0.912 & 0.956 & 0.958 & 0.989 & 0.946 & 0.976 \\ 
  & Max - Pooled & 0.934 & 0.971 & 0.961 & 0.990 & 0.943 & 0.978 \\ 
  & Wald - Stratified & 0.958 & 0.977 & 0.991 & 0.992 & 0.956 & 0.973 \\ 
  & Wald - Pooled & 0.949 & 0.975 & 0.994 & 0.995 & 0.972 & 0.979 \\ 
  & Logrank & 0.101 & 0.385 & 0.904 & 0.723 & 0.911 & 0.877 \\ 
  & Wilcoxon & 0.216 & 0.127 & 0.991 & 0.991 & 0.988 & 0.991 \\ 
  & Yang and Prentice & 0.840 & 0.974 & 0.982 & 0.974 & 0.961 & 0.951 \\ 
  & Two-Stage & 0.989 & 0.997 & 0.987 & 0.985 & 0.930 & 0.941 \\ 
  \hline
\end{tabular}
\end{table}

\begin{table}[ht]
\caption{\label{tab:crossing_unbalanced}Simulation results for the unbalanced crossing scenario. The column names specify the location of the crossing as well as the event rate. Shown is the rejection rate.}
\centering
\begin{tabular}{llrrrrrr}
&& \multicolumn{2}{c}{Mid-Crossing} & \multicolumn{2}{c}{End-Crossing} & \multicolumn{2}{c}{DTA} \\ 
  \cmidrule(l){3-8}
  n & Test & 0.7 & 0.9 & 0.7 & 0.9 & 0.7 & 0.9\\
  \hline
  \multirow{8}{4em}{30} & Max - Stratified & 0.091 & 0.091 & 0.440 & 0.427 & 0.394 & 0.425 \\ 
  & Max - Pooled & 0.146 & 0.171 & 0.345 & 0.426 & 0.312 & 0.398 \\ 
  & Wald - Stratified & 0.171 & 0.156 & 0.478 & 0.439 & 0.406 & 0.424 \\ 
  & Wald - Pooled & 0.149 & 0.153 & 0.501 & 0.508 & 0.421 & 0.430 \\ 
  & Logrank & 0.038 & 0.074 & 0.359 & 0.281 & 0.348 & 0.350 \\ 
  & Wilcoxon & 0.064 & 0.056 & 0.491 & 0.495 & 0.447 & 0.474 \\ 
  & Yang and Prentice & 0.257 & 0.389 & 0.449 & 0.442 & 0.406 & 0.435 \\ 
  & Two-Stage & 0.389 & 0.453 & 0.378 & 0.436 & 0.330 & 0.347 \\ 
  \hline
  \multirow{8}{4em}{50} & Max - Stratified & 0.241 & 0.289 & 0.572 & 0.606 & 0.556 & 0.609 \\ 
  & Max - Pooled & 0.316 & 0.392 & 0.554 & 0.632 & 0.531 & 0.589 \\ 
  & Wald - Stratified & 0.383 & 0.369 & 0.670 & 0.647 & 0.580 & 0.588 \\ 
  & Wald - Pooled & 0.341 & 0.393 & 0.728 & 0.738 & 0.611 & 0.638 \\ 
  & Logrank & 0.047 & 0.097 & 0.492 & 0.368 & 0.525 & 0.471 \\ 
  & Wilcoxon & 0.089 & 0.073 & 0.717 & 0.700 & 0.632 & 0.666 \\ 
  & Yang and Prentice & 0.366 & 0.582 & 0.654 & 0.622 & 0.603 & 0.585 \\ 
  & Two-Stage & 0.586 & 0.671 & 0.604 & 0.664 & 0.484 & 0.516 \\ 
  \hline
  \multirow{8}{4em}{100} & Max - Stratified & 0.682 & 0.774 & 0.841 & 0.906 & 0.822 & 0.868 \\ 
  & Max - Pooled & 0.719 & 0.825 & 0.852 & 0.917 & 0.817 & 0.868 \\ 
  & Wald - Stratified & 0.768 & 0.814 & 0.931 & 0.936 & 0.846 & 0.862 \\ 
  & Wald - Pooled & 0.771 & 0.818 & 0.948 & 0.948 & 0.868 & 0.880 \\ 
  & Logrank & 0.055 & 0.196 & 0.750 & 0.556 & 0.778 & 0.699 \\ 
  & Wilcoxon & 0.124 & 0.079 & 0.941 & 0.941 & 0.907 & 0.925 \\ 
  & Yang and Prentice & 0.654 & 0.881 & 0.907 & 0.883 & 0.863 & 0.850 \\ 
  & Two-Stage & 0.920 & 0.958 & 0.900 & 0.936 & 0.791 & 0.793 \\ 
  \hline
  \multirow{8}{4em}{200} & Max - Stratified & 0.987 & 0.994 & 0.991 & 0.998 & 0.987 & 0.996 \\ 
  & Max - Pooled & 0.982 & 0.994 & 0.996 & 0.998 & 0.988 & 0.996 \\ 
  & Wald - Stratified & 0.991 & 0.995 & 0.999 & 1.000 & 0.991 & 0.996 \\ 
  & Wald - Pooled & 0.988 & 0.994 & 0.999 & 1.000 & 0.993 & 0.996 \\ 
  & Logrank & 0.080 & 0.435 & 0.965 & 0.818 & 0.957 & 0.929 \\ 
  & Wilcoxon & 0.233 & 0.134 & 0.998 & 0.999 & 0.996 & 0.998 \\ 
  & Yang and Prentice & 0.918 & 0.996 & 0.997 & 0.994 & 0.989 & 0.987 \\ 
  & Two-Stage & 0.998 & 1.000 & 0.999 & 0.999 & 0.979 & 0.991 \\ 
  \hline
\end{tabular}
\end{table}

\clearpage

\printbibliography

@article{dey2025inference,
  title={Inference on overlap index: with an application to cancer data},
  author={Dey, Raju and Bathke, Arne C and Kumar, Somesh},
  journal={The International Journal of Biostatistics},
  volume={21},
  number={2},
  pages={357--383},
  year={2025},
  publisher={De Gruyter}
}

@article{schein1982comparison,
  title={A comparison of combination chemotherapy and combined modality therapy for locally advanced gastric carcinoma},
  author={Schein, Philip S and Gastrointestinal Tumor Study Group},
  journal={Cancer},
  volume={49},
  number={9},
  pages={1771--1777},
  year={1982},
  publisher={Wiley Online Library}
}

@article{hong2019adjuvant,
  title={Adjuvant whole-brain radiation therapy compared with observation after local treatment of melanoma brain metastases: a multicenter, randomized phase III trial},
  author={Hong, Angela M and Fogarty, Gerald B and Dolven-Jacobsen, Kari and Burmeister, Bryan H and Lo, Serigne N and Haydu, Lauren E and Vardy, Janette L and Nowak, Anna K and Dhillon, Haryana M and Ahmed, Tasnia and others},
  journal={Journal of clinical oncology},
  volume={37},
  number={33},
  pages={3132--3141},
  year={2019},
  publisher={American Society of Clinical Oncology}
}

@article{wang2009estimation,
  title={Estimation of the area under ROC curve with censored data},
  author={Wang, Qihua and Yao, Lili and Lai, Peng},
  journal={Journal of statistical planning and inference},
  volume={139},
  number={3},
  pages={1033--1044},
  year={2009},
  publisher={Elsevier}
}

@article{parkinson2022testing,
  title={Testing for equality of distributions using the concept of (niche) overlap},
  author={Parkinson-Schwarz, Judith H and Bathke, Arne C},
  journal={Statistical Papers},
  volume={63},
  number={1},
  pages={225--242},
  year={2022},
  publisher={Springer}
}

@article{Beck2023Combining,
author = {Beck, Jonas and Langthaler, Patrick B. and Bathke, Arne C.},
title = {Combining stochastic tendency and distribution overlap towards improved nonparametric effect measures and inference},
journal = {Scandinavian Journal of Statistics},
volume = {52},
number = {3},
pages = {1138-1175},
doi = {https://doi.org/10.1111/sjos.12783},
url = {https://onlinelibrary.wiley.com/doi/abs/10.1111/sjos.12783},
eprint = {https://onlinelibrary.wiley.com/doi/pdf/10.1111/sjos.12783},
year = {2025}
}

@article{ParkinsonKutilKupplerJunkerTrutschnigBathke+2018,
url = {https://doi.org/10.1515/ijb-2017-0028},
title = {A Fast and Robust Way to Estimate Overlap of Niches, and Draw Inference},
author = {Judith H. Parkinson and Raoul Kutil and Jonas Kuppler and Robert R. Junker and Wolfgang Trutschnig and Arne C. Bathke},
pages = {20170028},
volume = {14},
number = {2},
journal = {The International Journal of Biostatistics},
doi = {doi:10.1515/ijb-2017-0028},
year = {2018},
lastchecked = {2023-09-22}
}

@article{Dobler2018,
  author  = {Dobler, Dennis and Pauly, Markus},
  title   = {Bootstrap- and permutation-based inference for the Mann--Whitney effect for right-censored and tied data},
  journal = {TEST},
  year    = {2018},
  volume  = {27},
  number  = {3},
  pages   = {639--658},
  doi     = {10.1007/s11749-017-0565-z},
  url     = {https://doi.org/10.1007/s11749-017-0565-z},
  issn    = {1863-8260}
}

@book{vdVW2023,
   author    = {van der Vaart, Aad W. and Wellner, Jon A.},
   title     = {Weak Convergence and Empirical Processes: With Applications to Statistics},
   series    = {Springer Series in Statistics},
   edition   = {2},
   publisher = {Springer},
   year      = {2023}
 }

@article{horiguchi25,
author = {Horiguchi, Miki and Tian, Lu and Kehl, Kenneth and Uno, Hajime},
year = {2025},
month = {10},
pages = {784-809},
title = {Assessing delayed treatment benefits of immunotherapy using long-term average hazard: a novel test/estimation approach},
volume = {31},
journal = {Lifetime Data Analysis},
doi = {10.1007/s10985-025-09671-0}
}

@article{royston24,
author = {Royston, Patrick and Parmar, Mahesh},
year = {2014},
month = {08},
pages = {314},
title = {An approach to trial design and analysis in the era of non-proportional hazards of the treatment effect},
volume = {15},
journal = {Trials},
doi = {10.1186/1745-6215-15-314}
}

@article{gasparyan20,
author = {Gasparyan, Boris and Folkvaljon, Folke and Bengtsson, Olof and Buenconsejo, Joan and Koch, Gary},
year = {2020},
month = {07},
pages = {580-611},
title = {Adjusted win ratio with stratification: Calculation methods and interpretation},
volume = {30},
journal = {Statistical Methods in Medical Research},
doi = {10.1177/0962280220942558}
}

@article{Wyupek2021APT,
  title={A permutation test for the two-sample right-censored model},
  author={Grzegorz Wyłupek},
  journal={Annals of the Institute of Statistical Mathematics},
  year={2021},
  volume={73},
  pages={1037 - 1061},
  url={https://api.semanticscholar.org/CorpusID:254233111}
}

@article{Huang02092022,
author = {Xinghui Huang and Jingjing Lyu and Yawen Hou and Zheng Chen},
title = {A nonparametric statistical method for two crossing survival curves},
journal = {Communications in Statistics - Simulation and Computation},
volume = {51},
number = {9},
pages = {5041--5050},
year = {2022},
publisher = {Taylor \& Francis},
doi = {10.1080/03610918.2020.1753075},


URL = { 
    
        https://doi.org/10.1080/03610918.2020.1753075
    
    

},
eprint = { 
    
        https://doi.org/10.1080/03610918.2020.1753075
    
    

}

}

@article{yangprentice09,
author = {Yang, Song and Prentice, Ross},
year = {2009},
month = {05},
pages = {30-8},
title = {Improved Logrank-Type Tests for Survival Data Using Adaptive Weights},
volume = {66},
journal = {Biometrics},
doi = {10.1111/j.1541-0420.2009.01243.x}
}

@article{cantagallo21,
author = {Cantagallo, Eva and De Backer, Mickaël and Kicinski, Michal and Ozenne, Brice and Collette, Laurence and Legrand, Catherine and Buyse, Marc and Péron, Julien},
title = {A new measure of treatment effect in clinical trials involving competing risks based on generalized pairwise comparisons},
journal = {Biometrical Journal},
volume = {63},
number = {2},
pages = {272-288},
doi = {https://doi.org/10.1002/bimj.201900354},
url = {https://onlinelibrary.wiley.com/doi/abs/10.1002/bimj.201900354},
eprint = {https://onlinelibrary.wiley.com/doi/pdf/10.1002/bimj.201900354},
year = {2021}
}

@article{linwang04,
author = {Lin, Xun and Wang, Hongkun},
title = {A New Testing Approach for Comparing the Overall Homogeneity of Survival Curves},
journal = {Biometrical Journal},
volume = {46},
number = {5},
pages = {489-496},
doi = {https://doi.org/10.1002/bimj.200310053},
url = {https://onlinelibrary.wiley.com/doi/abs/10.1002/bimj.200310053},
eprint = {https://onlinelibrary.wiley.com/doi/pdf/10.1002/bimj.200310053},
year = {2004}
}

@article{Ananthakrishnan2021CriticalRO,
  title={Critical Review of Oncology Clinical Trial Design Under Non-proportional Hazards.},
  author={Revathi Ananthakrishnan and Stephanie Green and Alessandro Previtali and Rong Liu and Daniel H. Li and Michael P. Lavalley},
  journal={Critical reviews in oncology/hematology},
  year={2021},
  pages={
          103350
        },
  url={https://api.semanticscholar.org/CorpusID:234597228}
}

@article{linxu10,
author = {Lin, Xun and Xu, Qiang},
title = {A new method for the comparison of survival distributions},
journal = {Pharmaceutical Statistics},
volume = {9},
number = {1},
pages = {67-76},
doi = {https://doi.org/10.1002/pst.376},
url = {https://onlinelibrary.wiley.com/doi/abs/10.1002/pst.376},
eprint = {https://onlinelibrary.wiley.com/doi/pdf/10.1002/pst.376},
year = {2010}
}

@article{LiuYin2017Partitioned,
  author  = {Liu, Yukun and Yin, Guosheng},
  title   = {Partitioned log-rank tests for the overall homogeneity of hazard rate functions},
  journal = {Lifetime Data Analysis},
  year    = {2017},
  volume  = {23},
  number  = {3},
  pages   = {400--425},
  doi     = {10.1007/s10985-016-9365-0},
  pmid    = {26995734}
}

@article{lin23,
author = {Lin, Ray and Mukhopadhyay, Pralay and Roychoudhury, Satrajit and Anderson, Keaven and Hu, Tianle and Huang, Bo and Leon, Larry and Liao, Jason and Lin, Ji and Liu, Rong and Luo, Xiaodong and Mai, Yabing and Qin, Rui and Tatsuoka, Kay and Wang, Yang and Ye, Jiabu and Zhu, Jian and Chen, Tai-Tsang and Iacona, Renee},
year = {2023},
month = {05},
pages = {312-314},
title = {Comment on “Non-Proportional Hazards – an Evaluation of the MaxCombo Test in Cancer Clinical Trials” by the Cross-Pharma Non-Proportional Hazards Working Group},
volume = {15},
journal = {Statistics in Biopharmaceutical Research},
doi = {10.1080/19466315.2022.2103180}
}

@article{monnickendam19,
author = {Monnickendam, Giles and Zhu, Mingshu and McKendrick, Jan and Su, Yun},
year = {2019},
month = {04},
pages = {431-438},
title = {Measuring Survival Benefit in Health Technology Assessment in the Presence of Nonproportional Hazards},
volume = {22},
journal = {Value in Health},
doi = {10.1016/j.jval.2019.01.005}
}

@article{mao24,
author = {Mao, Lu},
year = {2024},
month = {07},
pages = {17407745241259356},
title = {Defining estimand for the win ratio: Separate the true effect from censoring},
volume = {21},
journal = {Clinical trials (London, England)},
doi = {10.1177/17407745241259356}
}

@article {Dehbij2250,
	author = {Dehbi, Hakim-Moulay and Royston, Patrick and Hackshaw, Allan},
	title = {Life expectancy difference and life expectancy ratio: two measures of treatment effects in randomised trials with non-proportional hazards},
	volume = {357},
	elocation-id = {j2250},
	year = {2017},
	doi = {10.1136/bmj.j2250},
	publisher = {BMJ Publishing Group Ltd},
	issn = {0959-8138},
	URL = {https://www.bmj.com/content/357/bmj.j2250},
	eprint = {https://www.bmj.com/content/357/bmj.j2250.full.pdf},
	journal = {BMJ}
}

@article{karrison16,
author = {Karrison, Theodore},
year = {2016},
month = {09},
pages = {678-690},
title = {Versatile Tests for Comparing Survival Curves Based on Weighted Log-rank Statistics},
volume = {16},
journal = {The Stata Journal: Promoting communications on statistics and Stata},
doi = {10.1177/1536867X1601600308}
}

@article{pak17,
author = {Pak, Kyongsun and Uno, Hajime and Kim, Dae Hyun and Tian, Lu and Kane, Robert and Takeuchi, Masahiro and Fu, Haoda and Claggett, Brian and Wei, Leejen},
year = {2017},
month = {12},
pages = {1692-1696},
title = {Interpretability of Cancer Clinical Trial Results Using Restricted Mean Survival Time as an Alternative to the Hazard Ratio},
volume = {3},
journal = {JAMA Oncology},
doi = {10.1001/jamaoncol.2017.2797}
}

@article{mukhopadyay22,
author = {Mukhopadhyay, Pralay and Ye, Jiabu and Anderson, Keaven and Roychoudhury, Satrajit and Rubin, Eric and Halabi, Susan and Chappell, Richard},
year = {2022},
month = {09},
pages = {1294-1300},
title = {Log-Rank Test vs MaxCombo and Difference in Restricted Mean Survival Time Tests for Comparing Survival Under Nonproportional Hazards in Immuno-oncology Trials: A Systematic Review and Meta-analysis},
volume = {8},
journal = {JAMA Oncology},
doi = {10.1001/jamaoncol.2022.2666}
}

@article{wugilbert22,
 ISSN = {0006341X, 15410420},
 URL = {http://www.jstor.org/stable/3068543},
 author = {Lang Wu and Peter B. Gilbert},
 journal = {Biometrics},
 number = {4},
 pages = {997--1004},
 publisher = {[Wiley, International Biometric Society]},
 title = {Flexible Weighted Log-Rank Tests Optimal for Detecting Early and/or Late Survival Differences},
 urldate = {2026-06-01},
 volume = {58},
 year = {2002}
}

@article{zhang24,
author = {Zhang, Xiaoxi and Datta, Somnath and Qiu, Peihua},
year = {2024},
month = {07},
pages = {1-21},
title = {Effective comparison of two potentially crossing hazard rate curves},
volume = {94},
journal = {Journal of Statistical Computation and Simulation},
doi = {10.1080/00949655.2024.2372631}
}

@article{royston13,
author = {Royston, Patrick and Parmar, Mahesh},
year = {2013},
month = {12},
pages = {152},
title = {Restricted mean survival time: An alternative to the hazard ratio for the design and analysis of randomized trials with a time-to-event outcome},
volume = {13},
journal = {BMC medical research methodology},
doi = {10.1186/1471-2288-13-152}
}

@article{liu20,
author = {Liu, Tiantian and Ditzhaus, Marc and Xu, Jin},
year = {2020},
month = {01},
pages = {},
title = {A resampling‐based test for two crossing survival curves},
volume = {19},
journal = {Pharmaceutical Statistics},
doi = {10.1002/pst.2000}
}

@article{dormuth22,
author = {Dormuth, Ina and Liu, Tiantian and Xu, Jin and Yu, Menggang and Pauly, Markus and Ditzhaus, Marc},
year = {2022},
month = {01},
pages = {},
title = {Which test for crossing survival curves? A user’s guideline},
volume = {22},
journal = {BMC Medical Research Methodology},
doi = {10.1186/s12874-022-01520-0}
}

@article{qiusheng08,
author = {Qiu, Peihua and Sheng, Jun},
year = {2008},
month = {02},
pages = {191-208},
title = {A Two-Stage Procedure for Comparing Hazard Rate Functions},
volume = {70},
journal = {Journal of the Royal Statistical Society Series B},
doi = {10.1111/j.1467-9868.2007.00622.x}
}

@article{Wu2026WinRA,
  title={Win Ratio as an Effect Size Measure Under Non‐Proportional Hazards: A Comparison With Difference in Restricted Mean Survival},
  author={Yuan Wu and Xiaofei Wang and Zhiguo Li},
  journal={Statistics in Medicine},
  year={2026},
  url={https://api.semanticscholar.org/CorpusID:287809501}
}

@article{you23,
author = {You, Na and He, Xueyi and Dai, Hongsheng and Wang, Xueqin},
title = {Ball divergence for the equality test of crossing survival curves},
journal = {Statistics in Medicine},
volume = {42},
number = {29},
pages = {5353-5368},
doi = {https://doi.org/10.1002/sim.9914},
url = {https://onlinelibrary.wiley.com/doi/abs/10.1002/sim.9914},
eprint = {https://onlinelibrary.wiley.com/doi/pdf/10.1002/sim.9914},
year = {2023}
}

@article{flandre20,
author = {Flandre, Philippe and O'Quigley, John},
title = {Note on the role of the placebo group in the short-term and long-term hazard ratio model},
journal = {Statistics in Medicine},
volume = {39},
number = {20},
pages = {2685-2688},
doi = {https://doi.org/10.1002/sim.8424},
url = {https://onlinelibrary.wiley.com/doi/abs/10.1002/sim.8424},
eprint = {https://onlinelibrary.wiley.com/doi/pdf/10.1002/sim.8424},
year = {2020}
}

@article{Li2015StatisticalIM,
  title={Statistical Inference Methods for Two Crossing Survival Curves: A Comparison of Methods},
  author={Huimin Li and Dong Han and Yawen Hou and Huilin Chen and Zheng Chen},
  journal={PLoS ONE},
  year={2015},
  volume={10},
  url={https://api.semanticscholar.org/CorpusID:18400244}
}

@Manual{rsurv,
    title = {rsurv: Random Generation of Survival Data},
    author = {Fabio Demarqui},
    year = {2024},
    note = {R package version 0.0.2},
    url = {https://CRAN.R-project.org/package=rsurv},
    doi = {10.32614/CRAN.package.rsurv},
}

@article{yang2005semiparametric,
  title={Semiparametric analysis of short-term and long-term hazard ratios with two-sample survival data},
  author={Yang, Song and Prentice, Ross},
  journal={Biometrika},
  volume={92},
  number={1},
  pages={1--17},
  year={2005},
  publisher={Oxford University Press}
}

\end{document}